\documentclass[a4paper,12pt,centertags,reqno,psamsfonts]{amsart}

\usepackage{comment}
\usepackage[headings]{fullpage}
\usepackage{kerkis}
\usepackage{letltxmacro}
\usepackage{setspace}
\usepackage[inline]{enumitem}
\usepackage{pifont}
\usepackage{array}
\usepackage{fancybox}
\usepackage{hyperref}
\usepackage{ifthen}
\usepackage{graphicx}
\usepackage{subfigure}
\usepackage{amsmath}
\usepackage{amsthm}
\usepackage{amssymb}				
\usepackage{mathtools}
\usepackage{mathrsfs}
\usepackage{stmaryrd}
\usepackage{yhmath}
\usepackage{accents}
\usepackage{xfrac}
\usepackage[all]{xy}
\usepackage[round,comma,authoryear,sort&compress]{natbib}
\usepackage{float}

\graphicspath{{Figures/}}

\setlist[itemize]{leftmargin=*,topsep=1ex,itemsep=0ex}
\setlist[enumerate]{leftmargin=*,topsep=1ex,itemsep=0ex}
\DeclareMathOperator{\Tr}{Tr}

\theoremstyle{plain}
\newtheorem{theorem}{Theorem}
\newtheorem{lemma}[theorem]{Lemma}
\newtheorem{proposition}[theorem]{Proposition}
\newtheorem{corollary}[theorem]{Corollary}
\theoremstyle{definition}
\newtheorem{definition}[theorem]{Definition}

\theoremstyle{remark}

\numberwithin{theorem}{section}
\numberwithin{equation}{section}
\numberwithin{figure}{section}
\numberwithin{table}{section}

\bibpunct{(}{)}{,}{a}{}{,}

\allowdisplaybreaks[4]

\makeatletter
\let\oldr@@t\r@@t
\def\r@@t#1#2{%
\setbox0=\hbox{$\oldr@@t#1{#2\,}$}\dimen0=\ht0
\advance\dimen0-0.2\ht0
\setbox2=\hbox{\vrule height\ht0 depth -\dimen0}%
{\box0\lower0.4pt\box2}}
\LetLtxMacro{\oldsqrt}{\sqrt}
\renewcommand*{\sqrt}[2][\ ]{\oldsqrt[#1]{#2}}
\makeatother

\newcommand{\N}{\mathbb{N}}

\newcommand{\R}{\mathbb{R}}

\usepackage[export]{adjustbox}

\newcommand{\cmpl}[1]{#1^\mathsf{c}}

\newcommand{\fnc}[3]{#1:#2\rightarrow #3}

\newcommand{\Cpq}[2]{
\ifthenelse{#1=0 \and #2=0}{\mathrm{C}}
{\ifthenelse{#2=0}{\mathrm{C}^{#1}}
{\mathrm{C}^{#1,#2}}}
}

\renewcommand{\d}{\mathrm{d}}

\newcommand{\SgAlg}[1]{\mathscr{#1}}

\newcommand{\Borel}{\mathscr{B}}

\newcommand{\Fltrn}[1]{\mathfrak{#1}}

\newcommand{\E}{\textsf{\upshape E}}

\renewcommand{\P}{\textsf{\upshape P}}
\renewcommand{\Pr}[1]{\textsf{\upshape{#1}}}
\newcommand{\ind}[1]{\mathbf{1}_{#1}}

\newcommand{\M}{\mathscr{M}}

\begin{document}
\title{Non-Parametric Robust Model Risk Measurement with Path-Dependent Loss Functions}

\author{Yu Feng}

\address{Yu Feng\\
Finance Discipline Group\\
University of Technology Sydney\\
P.O. Box 123\\
Broadway, NSW 2007\\
Australia}
\email{yu.feng-5@student.uts.edu.au}


\date{\today}

\begin{abstract}
Understanding and measuring model risk is important to financial practitioners. 
However, there lacks a non-parametric approach to model risk quantification in a dynamic setting and with path-dependent losses. 
We propose a complete theory generalizing the relative-entropic approach by \cite{glasserman2014robust} to the dynamic case under any $f$-divergence. It provides an unified treatment for measuring both the worst-case risk and the $f$-divergence budget that originate from the model uncertainty of an underlying state process. 
\nocite{Con16,CF10a,CF10b,CF13,Dup09,GS89,GX14,HS01,HS15,HSTW06,JS17,Sap17}
\end{abstract}

\maketitle

\section{Introduction}
As a working definition, model risk refers to the quantification of unanticipated losses resulting from the use of inappropriate models to value and manage financial securities, including widely traded securities like stocks and bonds, for which market prices are readily available, and less traded derivatives written on such securities. Unlike other financial risks, which are concerned with the impact of randomness within the paradigm of a chosen model, model risk is concerned with the possibility that the wrong modelling paradigm was chosen in the first place. This makes it a much more challenging proposition, both conceptually and in terms of implementation. It is thus unsurprising that model risk continues to languish behind its more traditional counterparts, such as price risk, interest rate risk and credit risk, both in terms of identifying an appropriate theoretical methodology and in the development of specific metrics.

A simple approach of accounting for model uncertainty is to assign weights to alternative models and then calculate the average market risk \cite[]{branger2004model}.
Perhaps a better way is to separate the model risk component from the market risk component. In addition, from the risk management point of view, one may be more interested in the worst-case scenario instead of the average scenario. \cite{kerkhof2002model} proposed a risk-differencing measure that separates the market risk under the worst-case model from the nominal market risk. 
Following the worst-case approach, \cite{cont2006model} formulated a quantitative framework for measuring the model risk in derivative pricing. This approach applies to a parametric set of alternative measures which price some benchmark instruments within their respective bid-ask spreads. 
Following Cont's work, \cite{gupta2010model} proposed the definition of the spread of a contingent claim to be the set of the prices given by all legitimate models. 
\cite{Ban&Sch:13} proposed a parametric risk framework that unifies the proposals of \cite{cont2006model}, \cite{gupta2010model} and \cite{lindstrom2010implications}. This approach incorporates a distribution of parameter values to capture the risk of parameter uncertainty, resulting in bid-ask spreads in instruments that face parameter risk. 
\cite{Det&Pac:16} approach the problem of model risk measurement based on the residual profit and loss from hedging in the reference model. 
\cite{Ker&Ber&Sch:10} propose a procedure to take model risk into account when computing capital reserves. Instead of formulating model risk in terms of a collection of probability measures, they consider the reality that practitioners may evaluate risk based on models of different natures. 
From a practical point of view, \cite{Bou&Dan&Kou&Mai:14} proposed an approach that incorporates model risk into the usual market risk measures. 

The approaches described above are parametric in the sense that they consider alternative models parametrised by a finite set of parameters. To go beyond that, \cite{glasserman2014robust} proposed a non-parametric approach. Under this framework, a worst-case model is found among alternative models in a neighborhood of a reference model. Glasserman and Xu adopted the relative entropy (or the Kullback-Leibler divergence) to measure the distance between the probability measure given by the reference model and an (equivalent) alternative measure. By imposing a constraint on the relative entropy budget, the set of legitimate alternative models is defined in a non-parametric fashion, and the worst-case scenario can then be solved analytically within a finite distance to the reference model. This approach is formulated w.r.t the distribution of a state variable, thus less applicable when the state variable evolves dynamically. In this paper, we apply it conceptually to the problem of measuring model risk w.r.t a state process. We solve the problem in a dual formulation and handle its path-dependency with the help of the functional Ito calculus \cite[]{Con16}. The constraint that defines the legitimate alternative models is w.r.t the $f$-divergence, a more general choice than the Kullback-Leibler divergence.
\section{Problem Formulation}
Fix $T\in(0,\infty)$ and $d\in\N$, and let $\Omega\coloneqq D([0,T],\R^d)$ denote the set of c\`adl\`ag paths $\fnc{\omega}{[0,T]}{\R^d}$. Let $[0,T]\ni t\mapsto X(t)$ be the canonical process on $\Omega$, which means to say that $X(t)(\omega)\coloneqq\omega(t)$, for all $(t,\omega)\in[0,T]\times\Omega$. Let $\Fltrn{F}^0=(\SgAlg{F}^0_t)_{t\in[0,T]}$ denote the filtration on $\Omega$ generated by $X$, which is to say that
\begin{equation*}
\SgAlg{F}^0_t\coloneqq\bigvee_{s\in[0,t]}\bigl\{X(s)^{-1}(U)\,\bigl|\,U\in\Borel(\R^d)\bigr\}
=\bigvee_{s\in[0,t]}\bigcup_{U\in\Borel(\R^d)}\{\omega\in\Omega\,|\,\omega(s)\in U\},
\end{equation*}
for all $t\in[0,T]$. In particular,
\begin{equation*}
\SgAlg{F}^0_0\coloneqq\bigl\{X(0)^{-1}(U)\,\bigl|\,U\in\Borel(\R^d)\bigr\}
=\bigcup_{U\in\Borel(\R^d)}\{\omega\in\Omega\,|\,\omega(0)\in U\}.
\end{equation*}
Fix a reference probability measure $\P$ on $(\Omega,\SgAlg{F}^0_T)$, subject to the condition
\begin{equation*}
\P\bigl(X(0)^{-1}(U)\bigr)=\P(\{\omega\in\Omega\,|\,\omega(0)\in U\})=
\begin{cases}
1&\text{if $0\in U$};\\
0&\text{if $0\notin U$},
\end{cases}	
\end{equation*}
for all $U\in\Borel(\R^d)$, which is to say that almost all paths start at zero under $\P$. Note that this condition ensures that $\P(A)=0$ or $\P(A)=1$, for all $A\in\SgAlg{F}^0_0$.

To be consistent with the notation in \citet{Con16}, we shall write $\omega_t\coloneqq\omega(t\wedge\cdot)\in\Omega$ to denote the path $\omega\in\Omega$ stopped at time $t\in[0,T]$. We impose an equivalence relation $\sim$ on $[0,T]\times\Omega$, by specifying that
\begin{equation*}
(t,\omega)\sim(t',\omega')\qquad\text{if and only if}\qquad t=t'\quad\text{and}\quad\omega_t=\omega'_{t'},
\end{equation*}
for all $(t,\omega),(t',\omega')\in[0,T]\times\Omega$. That is to say, two pairs, each consisting of a time and a path, are equivalent if the times are equal and the corresponding stopped paths are the same. The quotient set $\Lambda_T^d\coloneqq[0,T]\times\Omega\,/\!\sim$ forms a complete metric space, when endowed with the metric $\fnc{d_\infty}{(\Lambda_T^d)^2}{\R_+}$, defined by
\begin{equation*}
d_\infty\bigl((t,\omega),(t',\omega')\bigr)\coloneqq\sup_{s\in[0,T]}|\omega(s\wedge t)-\omega'(s\wedge t')|+|t-t'|
=\|\omega_t-\omega'_{t'}\|_\infty+|t-t'|,
\end{equation*}
for all $(t,\omega),(t',\omega')\in\Lambda_T^d$. We refer to $(\Lambda_T^d,d_\infty)$ as the \emph{space of stopped paths}.

A measurable function $\fnc{F}{\Lambda_T^d}{\R}$ is called a \emph{non-anticipative functional}, where $\Lambda_T^d$ is endowed with the Borel sigma-algebra generated by $d_\infty$ and $\R$ is endowed with the Borel sigma-algebra generated by the usual Euclidean metric. Since $(t,\omega)\sim(t,\omega_t)$, for all $(t,\omega)\in[0,T]\times\Omega$, we may regard a non-anticipative functional $\fnc{F}{\Lambda_T^d}{\R}$ as an appropriately measurable function $\fnc{F}{[0,T]\times\Omega}{\R}$ that satisfies the condition $F(t,\omega)=F(t,\omega_t)$. That is to say, the value of a non-anticipative functional, when applied to a particular time and path, depends only on the behaviour of the path up to the time. 
Note that $({F}(t,\,\cdot\,))_{t\in[0,T]}$ is a progressively measurable process, adapted to the filtration $\Fltrn{F}^0$.

Let $\M$ denote the family of (right-continuous versions of) martingales on the filtered probability space $(\Omega,\SgAlg{F}^0_T,\Fltrn{F}^0,\P)$, over the compact time-interval $[0,T]$, and let
\begin{equation*}
\M_+(1)\coloneqq\{Z\in\M\,|\,\text{$Z\geqslant 0$ and $Z(0)=1$}\}
\end{equation*}
denote the sub-family of non-negative martingales starting at one. Each $Z\in\M_+(1)$ defines a probability measure $\Pr{Q}_Z$ on $(\Omega,\SgAlg{F}^0_T)$ satisfying $\Pr{Q}_Z\ll\P$ (i.e. $\Pr{Q}_Z$ is absolutely continuous w.r.t $\P$), according to the recipe $\Pr{Q}_Z(A)\coloneqq\E\bigl(\ind{A}Z(T)\bigr)$, for all $A\in\SgAlg{F}^0_T$. Conversely, each probability measure $\Pr{Q}$ on $(\Omega,\SgAlg{F}^0_T)$ satisfying $\Pr{Q}\ll\P$ can be written as $\Pr{Q}=\Pr{Q}_Z$, where $Z\in\M_+(1)$ is determined by
\begin{equation*}
Z(t)\coloneqq\E\biggl(\frac{\d\Pr{Q}}{\d\P}\,\biggl|\,\SgAlg{F}^0_t\biggr), 	
\end{equation*}
for all $t\in[0,T]$.

Consider a twice-differentiable strictly convex function $\fnc{f}{\R_+}{\R}$ satisfying $f(1)=0$. 
For any probability measure $\Pr{Q}$ on $(\Omega,\SgAlg{F}^0_T)$ satisfying $\Pr{Q}\ll\P$, the $f$-divergence of $\Pr{Q}$ with respect to $\P$ is defined by
\begin{equation}\label{eq:defdf}
D_f(\Pr{Q}\|\P)\coloneqq\E\biggl(f\biggl(\frac{\d\Pr{Q}}{\d\P}\biggr)\biggr)
\end{equation}
\citep[see][Section~2]{Bas13}. Intuitively, $f$-divergence provides a measure of the distance between two probability measures. Hence, the set
\begin{equation*}
\mathcal{Z}_\eta\coloneqq\{Z\in\M_+(1)\,|\,D_f(\Pr{Q}_Z\|\P)\leqslant\eta\},
\end{equation*}
where $\eta\geqslant 0$, corresponds to the family of absolutely continuous probability measures that are close to the reference probability measure $\P$. 

Finally, fix a non-anticipative functional $\fnc{\ell}{\Lambda_T^d}{\R}$ satisfying $\ell(0,0)=0$. We shall interpret $\ell(t,\omega)$ as the cumulative realized loss  up to time $t$, incurred by a portfolio of financial securities. The state of the portfolio is completely determined by the path $\omega\in\Omega$. The condition of the reference probability measure guarantees 
\begin{equation*}
\P\bigl(X(0)^{-1}\{0\}\bigr)=\P(\{\omega\in\Omega\,|\,\omega(0)=0\})=1, 	
\end{equation*}
It follows that $\ell(0,\,\cdot\,)=0$ $\P$-a.s. That is to say, the initial realized loss incurred by the portfolio is zero under the reference probability measure. If we interpret $\P$ as the probability measure associated with a nominal model for the dynamics of the portfolio, then $\E\bigl(\ell(T,\,\cdot\,)\bigr)$ gives the expected total loss under the nominal model. In financial applications, we usually set the terminal time $T$ as the point when the entire portfolio gets liquidated, thus realizing the cumulative loss.

Suppose, now, that there is some uncertainty about which model best describes the portfolio. In particular, suppose that each probability measure determined by a member of $\mathcal{Z}_\eta$, for some $\eta\geqslant 0$, corresponds to a plausible model for the dynamics of the portfolio.\footnote{The idea here is that all absolutely continuous probability measures close enough to the reference measure (in the sense of $f$-divergence) correspond with models that are plausibly close to the reference model.} In that case, a risk manager would be interested in the following quantities:
\begin{equation}
\label{eqSec2:ModelRisk}
\sup_{Z\in\mathcal{Z}_\eta}\E^{\Pr{Q}_Z}\bigl(\ell(T,\,\cdot\,)\bigr)
\qquad\text{and}\qquad
\sup_{Z\in\mathcal{Z}_\eta}\E^{\Pr{Q}_Z}\bigl(\ell(T,\,\cdot\,)\bigr)-\E\bigl(\ell(T,\,\cdot\,)\bigr).
\end{equation}
The former expression may be regarded as the worst-case expected loss suffered by the portfolio under all plausible models, while the latter expression quantifies the difference between the worst-case expected loss and the expected loss under the default model. As such, it serves as a measure of model risk.

Problem defined in \eqref{eqSec2:ModelRisk} may be formulated in a dual form \cite[]{glasserman2014robust}. 
We first define the Lagrangian $\mathcal{L}:\M_+(1)\times(0,\infty)\times(0,\infty)\to\mathbb{R}$ by
\begin{equation*}
\mathcal{L}(Z,\vartheta, \eta)\coloneqq\E^{\Pr{Q}_Z}\bigl(\ell(T,\,\cdot\,)\bigr)-\frac{D_f(\Pr{Q}_Z\|\P)-\eta}{\vartheta}
=\E^{\Pr{Q}_Z}\biggl(\ell(T,\,\cdot\,)-\frac{f\bigl(Z(T)\bigr)}{\vartheta Z(T)}\biggr)+\frac{\eta}{\vartheta},
\end{equation*}
The Lagrangian leads to a dual function defined by
\begin{align*}
d(\vartheta, \eta)\coloneqq&\sup_{Z\in\M_+(1)}\mathcal{L}(Z,\vartheta, \eta)
=\sup_{Z\in\M_+(1)}\mathsf{E}^{\mathsf{Q}_Z}\bigl(\widehat\ell(T,Z)\bigr)+\frac{\eta}{\vartheta}
\end{align*}
Given $t\in[0,T]$ and $Z\in\M_+(1)$, 
\begin{equation}
\label{eqSec3:defhatl}
\widehat{\ell}_{\vartheta}(t,Z)\coloneqq\ell(t,\,\cdot\,)-\frac{f\bigl(Z(t)\bigr)}{\vartheta Z(t)}
\end{equation}
defines a $\SgAlg{F}^0_t$-measurable function $\widehat{\ell}_{\vartheta}(t,Z):\Omega\to\mathbb{R}$. 
As with $\ell:[0,T]\times\Omega\to\mathbb{R}$, $\widehat{\ell}_{\vartheta}( \cdot,Z)$ may be regarded as a non-anticipative functional.
 
If the primal problem is convex and the constraint satisfies Slater's condition \cite[]{slater2014lagrange}, then strong duality holds, giving
\begin{align}\label{eq:strong}
\sup_{Z\in\mathcal{Z}_\eta}\mathsf{E}^{\mathsf{Q}_Z}\left(\ell(T,\cdot)\right)=\inf_{\vartheta\in(0,\infty)}d(\vartheta, \eta)
\end{align}
This is proved in the following lemma.
\begin{lemma}\label{le:0}
The following statements are true:
\begin{enumerate}
\item
The set $\mathcal{Z}_\eta$ is convex.
\item
The function $\mathcal{Z}_\eta\ni Z\mapsto\E^{\Pr{Q}_z}\bigl(\ell(T,\,\cdot\,)\bigr)$ is convex.
\item
Strong duality Eq.~\ref{eq:strong} holds.
\item
Given $\vartheta^*\in(0,\infty)$, and suppose that $Z^*\in\M_+(1)$ satisfies
\begin{align*}
Z^*=\arg\max_{Z\in\M_+(1)}
\mathsf{E}^{\mathsf{Q}_Z}\bigl(\widehat\ell_{\vartheta^*}(T,Z)\bigr)
\end{align*}
then 
\begin{align*}
Z^*=\arg\max_{Z\in\mathcal{Z}_\eta}\E^{\mathsf{Q}_Z}\bigl(\ell(T,\cdot)\bigr)
\end{align*}
with $\eta:=\E\left(f\left(Z^*(T)\right)\right)$. 
\end{enumerate}
\end{lemma}
\begin{proof}
(1)~Given $Z_1,Z_2\in\mathcal{Z}_\eta$, observe that
\begin{align*}\label{eq:convex}
D_f(\Pr{Q}_{\lambda Z_1+(1-\lambda)Z_2}\|\P)
=D_f\bigl(\lambda\Pr{Q}_{Z_1}+(1-\lambda)\Pr{Q}_{Z_2}\bigl\|\P\bigr)
&=\E\bigl(f\bigl(\lambda Z_1(T)+(1-\lambda)Z_2(T)\bigr)\bigr)\\
&\leqslant\E\bigl(\lambda f\bigl(Z_1(T)\bigr)+(1-\lambda) f\bigl(Z_2(T)\bigr)\bigr)\\
&=\lambda\E\bigl(f\bigl(Z_1(T)\bigr)\bigr)+(1-\lambda)\E\bigl(f\bigl(Z_2(T)\bigr)\bigr)\\
&=\lambda D_f(\Pr{Q}_{Z_1}\|\P)+(1-\lambda)D_f(\Pr{Q}_{Z_2}\|\P),
\end{align*}
for all $\lambda\in[0,1]$, by virtue of the convexity of $f$ and Jensen's inequality. Since $D_f(\Pr{Q}_{Z_1}\|\P)\leqslant\eta$ and $D_f(\Pr{Q}_{Z_1}\|\P)\leqslant\eta$, the inequality above leads to $D_f(\Pr{Q}_{\lambda Z_1+(1-\lambda)Z_2}\|\P)\leqslant\eta$. This implies that $\lambda Z_1+(1-\lambda)Z_2\in\mathcal{Z}_\eta$, by virtue of the fact that $\lambda Z_1+(1-\lambda)Z_2\in\M_+(1)$.
\newline\vspace{2mm}\noindent
(2)~Given $Z_1,Z_2\in\mathcal{Z}_\eta$, observe that
\begin{align*}
\E^{\Pr{Q}_{\lambda Z_1+(1-\lambda)Z_2}}\bigl(\ell(T,\,\cdot\,)\bigr)	
&=\E\bigl(\bigl(\lambda Z_1(T)+(1-\lambda)Z_2(T)\bigr)\ell(T,\,\cdot\,)\bigr)\\
&=\lambda\E\bigl(Z_1(T)\ell(T,\,\cdot\,)\bigr)+(1-\lambda)\E\bigl(Z_2(T)\ell(T,\,\cdot\,)\bigr)\\
&=\lambda\E^{\Pr{Q}_{Z_1}}\bigl(\ell(T,\,\cdot\,)\bigr)+(1-\lambda)\E^{\Pr{Q}_{Z_2}}\bigl(\ell(T,\,\cdot\,)\bigr),
\end{align*}
for all $\lambda\in[0,1]$. 
Hence, the function $\mathcal{Z}_\eta\ni Z\mapsto\E^{\Pr{Q}_Z}\bigl(\ell(T,\,\cdot\,)\bigr)$ is linear and therefore also convex.\\
(3)~For a given $\eta\in(0,\infty)$, the constant process $Z=1$ satisfies $D_f(\mathsf{Q}_Z||P)=D_f(P||P)=0<\eta$. It is also an interior point of the subset $\mathcal{Z}_\eta\subseteq\M_+(1)$.\footnote{To see this point, consider the continuous function $\mathcal{H}:\M_+(1)\to\mathbb{R}$ defined by $\mathcal{H}(Z)=\E\left(f\left(Z(T)\right)\right)$ (we endow $\M_+(1)$ with the topology induced by the metric $d(Z_1,Z_2)=\E(|f(Z_1(T))-f_2(Z_2(T))|)$. The continuity ensures that $S:=\mathcal{H}^{-1}\left((-\eta, \eta)\right)$ is an open subset of $\M_+(1)$. Furthermore,
\begin{align*}
S\subseteq\{Z\in\M_+(1)|D_f(\Pr{Q}_Z||\Pr{P})<\eta\}\subseteq\mathcal{Z}_\eta
\end{align*}
suggesting that $S\subseteq \mathsf{int}(\mathcal{Z}_\eta)$. As an element in $S$, the constant process $Z=1$ is an interior point of $\mathcal{Z}_\eta$.}
According to Slater's condition \cite[]{slater2014lagrange}, the strong duality holds.
\\
(4)~Let $\eta:=\E\left(f\left(Z^*(T)\right)\right)$, and observe that
\begin{align*}
\inf_{\vartheta\in(0,\infty)}d(\vartheta, \eta)\leqslant&\, d(\vartheta^*, \eta)\\
=&\,\sup_{Z\in\M_+(1)} 
\mathsf{E}^{\mathsf{Q}_{Z}}\bigl(\widehat\ell_{\vartheta^*}(T,Z)\bigr)+\frac{\eta}{\vartheta^*}\\
=&\,\mathsf{E}^{\mathsf{Q}_{Z^*}}\bigl(\widehat\ell_{\vartheta^*}(T,Z^*)\bigr)+\frac{\eta}{\vartheta^*}\\
=&\,\mathsf{E}^{\mathsf{Q}_{Z^*}}\bigl(\ell(T,\cdot)\bigr)-\frac{1}{\vartheta^*}\E^{\mathsf{Q}_{Z^*}}\left(\frac{f(Z^*(T))}{Z^*(T)}\right)+\frac{\eta}{\vartheta^*}\\
=&\,\mathsf{E}^{\mathsf{Q}_{Z^*}}\bigl(\ell(T,\cdot)\bigr)-\frac{1}{\vartheta^*}\E\left({f(Z^*(T))}\right)+\frac{\eta}{\vartheta^*}\\
=&\,\mathsf{E}^{\mathsf{Q}_{Z^*}}\bigl(\ell(T,\cdot)\bigr)\\
\leqslant&\,
\sup_{Z\in\mathcal{Z}_\eta}\mathsf{E}^{\mathsf{Q}_{Z}}\bigl(\ell(T,\cdot)\bigr)
\end{align*}
Lemma \ref{le:0}(3) then ensures that
\begin{align*}
\inf_{\vartheta\in(0,\infty)}d(\vartheta, \eta)
=\mathsf{E}^{\mathsf{Q}_{Z^*}}\bigl(\ell(T,\cdot)\bigr)
=\sup_{Z\in\mathcal{Z}_\eta}\mathsf{E}^{\mathsf{Q}_{Z}}\bigl(\ell(T,\cdot)\bigr)
\end{align*}
and the result follows.
\end{proof}
For the primal problem formulated in Eq.~\ref{eqSec2:ModelRisk}, Lemma.~\ref{le:0}(4) implies the existence of a solution $Z^*$ that lies on the boundary of $\mathcal{Z}_\eta$ given $\eta>0$ (i.e. $\E\left(f\left(Z^*(T)\right)\right)=\eta$), as long as $Z^*$ solves 
\begin{equation}
\label{eqSec3:IntProb}
\max_{Z\in\M_+(1)}\E^{\Pr{Q}_Z}\bigl(\widehat{\ell}_{\vartheta}(T,Z)\bigr)
\end{equation}
for some $\vartheta\in(0,\infty)$. In the following context, we will consider the dual problem formulated in Eq.~\ref{eqSec3:IntProb} instead of the primal problem. For simplicity, we will regard $\theta>0$ as given and express $\widehat{\ell}_{\vartheta}$ by $\widehat{\ell}$.

\section{Characterising the Worst-Case Expected Loss}
This section provides implicit characterisation of the solution to the worst-case expected loss problem formulated in \eqref{eqSec2:ModelRisk}. 

Given $t\in[0,T]$ and $\bar{Z}\in\M_+(1)$, define the family of $\bar{Z}$-consistent martingale densities up to time $t$ by
\begin{equation*}
\mathcal{Z}(t,\bar{Z})\coloneqq\{Z\in\M_+(1)\,|\,Z(t)=\bar{Z}(t)\}.
\end{equation*}
Note that $\mathcal{Z}(0,\bar{Z})=\M_+(1)$, since $Z(0)=1=\bar{Z}(0)$ for all $Z\in\M_+(1)$. Note that the martingale property of the members of $\mathcal{Z}(t,\bar{Z})$ ensures that
\begin{equation*}
Z(s)=\E\bigl(Z(t)\,|\,\SgAlg{F}^0_s\bigr)=\E\bigl(\bar{Z}(t)\,|\,\SgAlg{F}^0_s\bigr)=\bar{Z}(s),
\end{equation*}
for all $Z\in\mathcal{Z}(t,\bar{Z})$ and all $s\in[0,t]$. In other words, $\mathcal{Z}(t,\bar{Z})$ is the set of processes in $\M_+(1)$ that are consistent with $\bar{Z}$ over the interval $[0,t]$. Moreover, we observe that
\begin{align*}
\Pr{Q}_{Z}(A)
=\E\bigl(\ind{A}Z(T)\bigr)
=\E\bigl(\E\bigl(\ind{A}Z(T)\,|\,\SgAlg{F}^0_t\bigr)\bigr)
=\E\bigl(\ind{A}Z(t)\bigr)
&=\E\bigl(\ind{A}\bar{Z}(t)\bigr)\\
&=\E\bigl(\E\bigl(\ind{A}\bar{Z}(T)\,|\,\SgAlg{F}^0_t\bigr)\bigr)\\
&=\E\bigl(\ind{A}\bar{Z}(T)\bigr)\\
&=\Pr{Q}_{\bar{Z}}(A),
\end{align*}
for all $Z\in\mathcal{Z}(t,\bar{Z})$ and all $A\in\SgAlg{F}^0_t$. That is to say, the probability measures associated with members of $\mathcal{Z}(t,\bar{Z})$ agree with each other on all $\SgAlg{F}^0_t$-measurable events. This is the set of feasible alternative measures by looking forward (from time $t$).

Given $\bar{Z}\in\M_+(1)$, we now define the $\Fltrn{F}^0$-adapted process $(\widehat{L}(t,\bar{Z}))_{t\in[0,T]}$ by 
\begin{align}
\label{eqSec3:CondIntProb}
\widehat{L}(t,\bar{Z})&\coloneqq\max_{Z\in\mathcal{Z}(t,\bar{Z})}\E^{\Pr{Q}_{Z}}\bigl(\widehat{\ell}(T,Z)-\widehat{\ell}(t,Z)\,\bigl|\,\SgAlg{F}^0_t\bigr)
\end{align}
for all $t\in[0,T]$, assuming the maximum always exists. Since $\widehat{\ell}(\,\cdot\,,Z)$ is a non-anticipative functional satisfying $\widehat{\ell}(0,Z)=0$ $\P$-a.s. and $Z(0)=1$ implies that $\Pr{Q}_Z|_{\SgAlg{F}^0_0}=\P|_{\SgAlg{F}^0_0}$, it follows that $\widehat{\ell}(0,Z)=0$ $\Pr{Q}_Z$-a.s. as well. Consequently,
\begin{equation}\label{eq:LoZ}
\widehat{L}(0,\bar{Z})=\max_{Z\in\M_+(1)}\E^{\Pr{Q}_{Z}}\bigl(\widehat{\ell}(T,Z)\,\bigl|\,\SgAlg{F}^0_0\bigr)
=\max_{Z\in\M_+(1)}\E^{\Pr{Q}_{Z}}\bigl(\widehat{\ell}(T,Z)\bigr),
\end{equation}
where the second equality follows from the fact that $\SgAlg{F}^0_0$ and $\SgAlg{F}^0_T$ are independent sigma-algebras, with respect to $\Pr{Q}_Z$.\footnote{First observe that $Z(0)=1$ implies that $\Pr{Q}_Z(A)=\P(A)=0$ or $\Pr{Q}_Z(A)=\P(A)=1$, for all $A\in\SgAlg{F}^0_0$. Consequently, given $A\in\SgAlg{F}^0_0$ and $B\in\SgAlg{F}^0_T$, we obtain
\begin{equation*}
0\leqslant\Pr{Q}_Z(A\cap B)\leqslant\Pr{Q}_Z(A)=0=\Pr{Q}_Z(A)\Pr{Q}_Z(B),
\end{equation*}
in the case when $\Pr{Q}_Z(A)=0$, while
\begin{align*}
\Pr{Q}_Z(A)\Pr{Q}_Z(B)=\Pr{Q}_Z(B)\geqslant\Pr{Q}_Z(A\cap B)
=\Pr{Q}_Z\bigl(\cmpl{(\cmpl{A}\cup\cmpl{B})}\bigr)
=1-\Pr{Q}_Z(\cmpl{A}\cup\cmpl{B})
&\geqslant 1-\bigl(\Pr{Q}_Z(\cmpl{A})+\Pr{Q}_Z(\cmpl{B})\bigr)\\
&=1-\Pr{Q}_Z(\cmpl{B})\\
&=\Pr{Q}_Z(B)\\
&=\Pr{Q}_Z(A)\Pr{Q}_Z(B),
\end{align*}
in the case when $\Pr{Q}_Z(A)=1$.} This is simply the problem given in Eq.~\ref{eqSec3:IntProb}.

\begin{definition} \label{def:z*}
A worst-case density process is some $Z^*\in\M_+(1)$ that solves the maximisation problem \eqref{eqSec3:CondIntProb} w.r.t the family of $Z^*$-consistent martingale densities:
\begin{equation}\label{eq:def0}
\E^{\Pr{Q}_{Z^*}}\bigl(\widehat{\ell}(T,Z^*)-\widehat{\ell}(t,Z^*)\,\bigl|\,\SgAlg{F}^0_t\bigr)=\widehat{L}(t,Z^*)
\end{equation}
for each $t\in[0,T]$.
\end{definition}
Suppose $Z^*\in\M_+(1)$ is a worst-case martingale density according to the definition above, then $Z^*$ solves the problem formulated in Eq.~\ref{eqSec3:IntProb}. This is confirmed by substituting Eq.~\ref{eq:LoZ} into Eq.~\ref{eq:def0} which leads to $
\E^{\Pr{Q}_{Z^*}}\bigl(\widehat{\ell}(T,Z^*)\bigr)=\max_{Z\in\M_+(1)}\E^{\Pr{Q}_{Z}}\bigl(\widehat{\ell}(T,Z)\bigr)$. 
In the proposition below, we characterizes such worst-case density by its martingale property.

\begin{proposition}\label{pro:1}
Fix $\bar{Z}\in\M_+(1)$ and suppose 
the maximum in \eqref{eqSec3:CondIntProb} exists
 for each $t\in[0,T]$. Then the process $[0,T]\ni t\mapsto\widehat{L}(t,\bar{Z})+\widehat{\ell}(t,\bar{Z})$ is a $\Pr{Q}_{\bar{Z}}$-supermartingale. 
It is a $\Pr{Q}_{\bar{Z}}$-martingale iff $\bar{Z}$ is a worst-case density process.
\end{proposition}
\begin{proof}
Given an arbitrary $t\in[0,T]$, we suppose $Z'\in\mathcal{Z}(t,\bar{Z})$ solves the maximisation problem (Eq.~\ref{eqSec3:CondIntProb}). 
Applying the law of iterated expectation, we have
\begin{align}\label{eq:iterated}
\E^{\Pr{Q}_{Z'}}\left(\left.\widehat\ell(T, Z')-\widehat\ell(t, Z')\,\right|\SgAlg{F}^0_s\right)
=&\E^{\Pr{Q}_{Z'}}\biggl(\biggl.\E^{\Pr{Q}_{Z'}}\left(\left.\widehat\ell(T, Z')-\widehat\ell(t, Z')\,\right|\SgAlg{F}^0_t\right)\biggr|\SgAlg{F}^0_s\biggr)\nonumber\\
=&\E^{\Pr{Q}_{Z'}}\left(\left.\widehat{L}(t,\bar{Z})\,\right|\SgAlg{F}^0_s\right)
\end{align}
for all $s\in[0, t]$. 
By virtue of $Z'(s)=\bar{Z}(s)$, $\widehat\ell(s, Z')=\widehat\ell(s, \bar{Z})$ for all $s\in[0, t]$. The same condition also leads to $Z'\in\mathcal{Z}(s,\bar{Z})$. 
According to the definition of $\widehat{L}$ (Eq.~\ref{eqSec3:CondIntProb}), we have the following inequality
\begin{align}\label{eq:jnequal}
\widehat{L}(s,\bar{Z})&\geqslant \E^{\Pr{Q}_{Z'}}\left(\left.\widehat\ell(T, Z')-\widehat\ell(s, Z')\,\right|\SgAlg{F}^0_s\right)\nonumber\\
&=\E^{\Pr{Q}_{Z'}}\left(\left.\widehat\ell(t, Z')-\widehat\ell(s, Z')\,\right|\SgAlg{F}^0_s\right)+
\E^{\Pr{Q}_{Z'}}\left(\left.\widehat\ell(T, Z')-\widehat\ell(t, Z')\,\right|\SgAlg{F}^0_s\right)\nonumber\\
&=\E^{\Pr{Q}_{\bar{Z}}}\left(\left.\widehat\ell(t, \bar{Z})-\widehat\ell(s, \bar{Z})+\widehat{L}(t,\bar{Z})\,\right|\SgAlg{F}^0_s\right)
\end{align}
for all $s\in[0, t]$. 
In the last equality, we replace $\Pr{Q}_{Z'}$ by $\Pr{Q}_{\bar{Z}}$ because $\widehat{\ell}(t,\bar{Z})$, $\widehat{\ell}(s,\bar{Z})$ and $\widehat{L}(t,\bar{Z})$ are all $\SgAlg{F}^0_t$-measurable.\footnote{The conditional expectation of a $\SgAlg{F}^0_t$-measurable function $X:\Omega\to\mathbb{R}$ w.r.t a sub-$\sigma$-algebra $\SgAlg{F}^0_s\subseteq\SgAlg{F}^0_t$ is
\begin{align*}
\E^{\Pr{Q}_{Z'}}\left(\left.X\,\right|\SgAlg{F}^0_s\right)
=\E\left(\left.\frac{Z'(T)}{Z'(s)}X\,\right|\SgAlg{F}^0_s\right)
=\E\left(\left.\frac{Z'(t)}{Z'(s)}\E\left(\left.\frac{Z'(T)}{Z'(t)}X\,\right|\SgAlg{F}^0_t\right)\,\right|\SgAlg{F}^0_s\right)
=&\,\E\left(\left.\frac{Z'(t)}{Z'(s)}\E^{\Pr{Q}_{Z'}}\left(\left.X\,\right|\SgAlg{F}^0_t\right)\,\right|\SgAlg{F}^0_s\right)\\
=&\,\E\left(\left.\frac{\bar{Z}(t)}{\bar{Z}(s)}X\,\right|\SgAlg{F}^0_s\right)\\
=&\,\E^{\Pr{Q}_{\bar{Z}}}\left(\left.X\,\right|\SgAlg{F}^0_s\right)
\end{align*}
}
Since $t\in[0,T]$ is chosen arbitrarily, Eq.~\ref{eq:jnequal} holds for any $s$ and $t$ that satisfies $0\leqslant s\leqslant t\leqslant T$. 

By re-arranging Eq.~\ref{eq:jnequal}, we obtain the supermartingale property of the $\Fltrn{F}^0$-adapted process $[0,T]\ni t\mapsto\widehat{L}(t,\bar{Z})+\widehat{\ell}(t,\bar{Z})$:
\begin{align}\label{eq:Ll}
\widehat{L}(s,\bar{Z})+\widehat\ell(s, \bar{Z})\geqslant \E^{\Pr{Q}_{\bar{Z}}}\left(\left.\widehat{L}(t,\bar{Z})+\widehat\ell(t, \bar{Z})\,\right|\SgAlg{F}^0_s\right)
\end{align}
The process is a $\Pr{Q}_{\bar{Z}}$-martingale iff the equality holds for all $0\leqslant s\leqslant t\leqslant T$. If $\bar{Z}$ is a worst-case density process, then according to 
Definition~\ref{def:z*} $\bar{Z}$ solves Eq.~\ref{eqSec3:CondIntProb} for all $t\in[0,T]$. We may set $Z'=\bar{Z}$ in Eq.~\ref{eq:jnequal} so that the first line takes the equal sign for all $s\in[0,t]$.  Conversely, if the equality holds for all $0\leqslant s\leqslant t\leqslant T$, then it holds for all $0\leqslant s\leqslant t=T$. 
By taking the equal sign in Eq.~\ref{eq:Ll} and replacing $t$ by $T$, we get
\begin{align*}
\widehat{L}(s,\bar{Z})
= \E^{\Pr{Q}_{\bar{Z}}}\left(\left.\widehat\ell(T, \bar{Z})-\widehat\ell(s, \bar{Z})\,\right|\SgAlg{F}^0_s\right)
\end{align*} 
for all $s\in[0,T]$, confirming that $\bar{Z}$ is a worst-case density process by 
Definition~\ref{def:z*}. 
\end{proof}
Proposition.~\ref{pro:1} can be regarded as generalization of the dynamic programming equation. In fact, given an optimal martingale density $Z^*\in\M_+(1)$, we take an arbitrary $\bar{Z}\in\mathcal{Z}(s,Z^*)$ and substitute it into Eq.~\ref{eq:Ll}. By observing that $\bar{Z}\in\mathcal{Z}(s,Z^*)$ matches $Z^*$ up to time $s$, we transform Eq.~\ref{eq:Ll} into
\begin{align*}
\widehat{L}(s,Z^*)+\widehat\ell(s, Z^*)\geqslant \E^{\Pr{Q}_{\bar{Z}}}\left(\left.\widehat{L}(t,\bar{Z})+\widehat\ell(t, \bar{Z})\,\right|\SgAlg{F}^0_s\right)
\end{align*}
The inequality holds for all $\bar{Z}\in\mathcal{Z}(s,Z^*)$. It takes the equal sign when $\bar{Z}=Z^*$. This leads to the following dynamic programming equation with respect to the density process,
\begin{align*}
\widehat{L}(s,Z^*)+\widehat\ell(s, Z^*)=\max_{Z\in\mathcal{Z}(s,Z^*)} \E^{\Pr{Q}_{Z^*}}\left(\left.\widehat{L}(t,Z)+\widehat\ell(t, Z)\,\right|\SgAlg{F}^0_s\right)
\end{align*}
for all $s$ and $t$ that satisfies $0\leqslant s\leqslant t\leqslant T$. 

\section{General Result of Model Risk Measurement}
We have shown in Proposition.~\ref{pro:1} that the $\Fltrn{F}^0$-adapted process $[0,T]\ni t\mapsto\widehat{L}(t,{Z}^*)+\widehat{\ell}(t,{Z}^*)$ is a $\Pr{Q}_{{Z}^*}$-martingale iff ${Z}^*$ is a worst-case density process. In this section, we will show that such $Z^*$ indeed exists under certain conditions and is characterized by an equation. This leads to a complete solution to the problem formulated in Eq.~\ref{eqSec2:ModelRisk}. First we prove a lemma.
\begin{lemma}\label{le:1}
Fix a martingale density $\bar{Z}\in\M_+(1)$. A measurable process $C:[0,T]\times\Omega\to\mathbb{R}$, satisfying
\begin{equation}\label{eq:le21}
\E^{\Pr{Q}_{Z}}\left(\left.C(t,\,\cdot\,)\right|\SgAlg{F}^0_t\right)\leqslant \E^{\Pr{Q}_{\bar{Z}}}\left(\left.C(t,\,\cdot\,)\right|\SgAlg{F}^0_t\right) 
\end{equation}
for all $t\in[0,T]$ and all $Z\in\mathcal{Z}(t,\bar{Z})$, admits a progressively measurable modification, i.e. there exists a progressively measurable process $\tilde{C}:[0,T]\times\Omega\to\mathbb{R}$, regarded as a non-anticipative functional, satisfying $\Pr{Q}_{\bar{Z}}\bigl(\{\omega\in\Omega\,|\,C(t,\omega)=\tilde{C}(t,\omega)\}\bigr)=1$ for every $t\in[0,T]$. 
\end{lemma}
\begin{proof}
The $\SgAlg{F}^0_t$-measurable function $u(t,\cdot):=\E^{\Pr{Q}_{\bar{Z}}}\bigl(C(t,\cdot)\,|\,\SgAlg{F}^0_t\bigr)$ 
forms a $\Fltrn{F}^0$-adapted process $\left(u(t,\cdot)\right)_{t\in[0,T]}$. It admits a progressively measurable modification $\bigl(\tilde{C}(t,\cdot)\bigr)_{t\in[0,T]}$ \cite[]{karatzas1991}. We would like to show that $\Pr{Q}_{\bar{Z}}\bigl(\{\omega\in\Omega\,|\,C(t,\omega)=\tilde{C}(t,\omega)\}\bigr)=1$ for every $t\in[0,T]$. 

We prove this lemma by contradiction. Suppose there exists a $t\in[0,T]$ such that $\Pr{Q}_{\bar{Z}}\bigl(\{\omega\in\Omega|C(t,\omega)=\tilde{C}(t,\omega)\}\bigr)<1$, then $\Pr{Q}_{\bar{Z}}\bigl(\{\omega\in\Omega\,|\,C(t,\omega)=u(t,\omega)\}\bigr)<1$.\footnote{We only need to prove $\Pr{Q}_{\bar{Z}}\bigl(\{\omega\in\Omega\,|\,C(t,\omega)={u}(t,\omega)\}\bigr)=1$ leads to $\Pr{Q}_{\bar{Z}}\bigl(\{\omega\in\Omega\,|\,C(t,\omega)=\tilde{C}(t,\omega)\}\bigr)=1$. In fact, assuming $\Pr{Q}_{\bar{Z}}\bigl(\{\omega\in\Omega\,|\,C(t,\omega)={u}(t,\omega)\}\bigr)=1$ we have
\begin{align*}
\Pr{Q}_{\bar{Z}}\bigl(\{\omega\in\Omega\,|\,C(t,\omega)=\tilde{C}(t,\omega)\}\bigr)
=&\,1-
\Pr{Q}_{\bar{Z}}\bigl(\{\omega\in\Omega\,|\,C(t,\omega)\neq\tilde{C}(t,\omega)\}\bigr)\\
\geqslant&\,
1-
\Pr{Q}_{\bar{Z}}\bigl(\{\omega\in\Omega\,|\,C(t,\omega)\neq u(t,\omega)\}\cup 
\{\omega\in\Omega\,|\,u(t,\omega)\neq\tilde{C}(t,\omega)\}\bigr)\\
\geqslant&\,
1-
\Pr{Q}_{\bar{Z}}\bigl(\{\omega\in\Omega\,|\,C(t,\omega)\neq u(t,\omega)\}\bigr)
-\Pr{Q}_{\bar{Z}}\bigl(\{\omega\in\Omega\,|\,u(t,\omega)\neq\tilde{C}(t,\omega)\}\bigr)\\
=&\,1
\end{align*}
}
This implies that either $\Pr{Q}_{\bar{Z}}\bigl(\{\omega\in\Omega\,|\,C(t,\omega)<{u}(t,\omega)\}\bigr)>0$ or $\Pr{Q}_{\bar{Z}}\bigl(\{\omega\in\Omega\,|\,C(t,\omega)>{u}(t,\omega)\}\bigr)>0$. Without losing generality, we assume $\Pr{Q}_{\bar{Z}}\bigl(\{\omega\in\Omega\,|\,C(t,\omega)<{u}(t,\omega)\}\bigr)>0$.

For notational simplicity, in the rest of the proof we use $C$ to denote the random variable $C(t,\cdot)$ and $u$ to denote the $\SgAlg{F}^0_t$-measurable function $u(t,\cdot)$. 
We construct an alternative martingale density $Z'\in\mathcal{Z}(t,{\bar{Z}})$ by
\begin{align}
Z'(s)=
	\begin{dcases}
	{\bar{Z}}(s)&~~~~~~~~~~~~s\in[0,t]\\
\frac{\E^{\Pr{Q}_{\bar{Z}}}\left(\left.e^C\ind{C<u}+e^u\ind{C\geq u}\,\right|\SgAlg{F}^0_s\right)}{\E^{\Pr{Q}_{\bar{Z}}}\left(\left.e^C\ind{C<u}+e^u\ind{C\geq u}\,\right|\SgAlg{F}^0_t\right)}{\bar{Z}}(s)
&~~~~~~~~~~~~s\in(t,T]
	\end{dcases}
\end{align}
To show that indeed $Z'\in\mathcal{Z}(t,{\bar{Z}})$, we need to prove that $Z'(0)=1$, $Z'\geqslant 0$, $Z'(t)={\bar{Z}}(t)$ and $Z'$ is a $\Pr{P}$-martingale. The first three conditions are obvious from the definition. 
The martingale property of $\left(Z'(s)\right)_{s\in[0,t]}$ is clear. The martingale property of $\left(Z'(s)\right)_{s\in[t,T]}$ is confirmed by 
\begin{align*}
\E\left(Z'(r)\,|\,\SgAlg{F}^0_s\right) =&\,
\E\left(\left.
\frac{\E\left(\left.{\bar{Z}}(T)\left(e^C\ind{C<u}+e^u\ind{C\geq u}\right)\,\right|\SgAlg{F}^0_r\right)}{\E^{\Pr{Q}_{{\bar{Z}}}}\left(\left.e^C\ind{C<u}+e^u\ind{C\geq u}\,\right|\SgAlg{F}^0_t\right)}\,\right|\SgAlg{F}^0_s\right)\\
=&\,
\frac{\E\left(\left.{\bar{Z}}(T)\left(e^C\ind{C<u}+e^u\ind{C\geq u}\right)\,\right|\SgAlg{F}^0_s\right)}{\E^{\Pr{Q}_{\bar{Z}}}\left(\left.e^C\ind{C<u}+e^u\ind{C\geq u}\,\right|\SgAlg{F}^0_t\right)}
=Z'(s)
\end{align*}
for all $s\in[t,T]$ and $r\in[t,T]$ satisfying $s\leqslant r$.

Because $\E^{\Pr{Q}_{\bar{Z}}}\bigl(\E^{\Pr{Q}_{\bar{Z}}}\left(\left.\ind{C<u}\,\right|\SgAlg{F}^0_t\right)\bigr)=\Pr{Q}_{\bar{Z}}(C<u)>0$, there exists a $\omega\in\Omega$ such that $\E^{\Pr{Q}_{\bar{Z}}}\left(\left.\ind{C<u}\,\right|\SgAlg{F}^0_t\right)(\omega)>0$. We define 
\begin{align*}
w_l\coloneqq&\,{\E^{\Pr{Q}_{\bar{Z}}}\left(\left.\ind{C<u}\,\right|\SgAlg{F}^0_t\right)(\omega)}\\
w_u\coloneqq&\,1-w_l=\E^{\Pr{Q}_{\bar{Z}}}\left(\left.\ind{C\geq u}\,\right|\SgAlg{F}^0_t\right)(\omega)\\
L\coloneqq&\,\ln\frac{\E^{\Pr{Q}_{\bar{Z}}}\left(\left.e^C\ind{C<u}\,\right|\SgAlg{F}^0_t\right)(\omega)}{\E^{\Pr{Q}_{\bar{Z}}}\left(\left.\ind{C<u}\,\right|\SgAlg{F}^0_t\right)(\omega)}\\
c_l\coloneqq&\,\frac{\E^{\Pr{Q}_{\bar{Z}}}\left(\left.Ce^C\ind{C< u}\,\right|\SgAlg{F}^0_t\right)(\omega)}{\E^{\Pr{Q}_{\bar{Z}}}\left(\left.e^C\ind{C< u}\,\right|\SgAlg{F}^0_t\right)(\omega)}\\
c_u\coloneqq&\,\frac{\E^{\Pr{Q}_{\bar{Z}}}\left(\left.C\ind{C\geq u}\,\right|\SgAlg{F}^0_t\right)(\omega)}{\E^{\Pr{Q}_{\bar{Z}}}\left(\left.\ind{C\geq u}\,\right|\SgAlg{F}^0_t\right)(\omega)}
\end{align*}
then the LHS of Eq.~\ref{eq:le21} (with $Z$ replaced by $Z'$) satisfies
\begin{samepage}
\begin{align}\label{eq:m0}
\E^{\Pr{Q}_{Z'}}\left(\left.C\,\right|\SgAlg{F}^0_t\right)(\omega)=
\frac{\E^{\Pr{Q}_{\bar{Z}}}\left(\left.Ce^C\ind{C<u}+Ce^u\ind{C\geq u}\,\right|\SgAlg{F}^0_t\right)(\omega)}{\E^{\Pr{Q}_{\bar{Z}}}\left(\left.e^C\ind{C<u}+e^u\ind{C\geq u}\,\right|\SgAlg{F}^0_t\right)(\omega)}
=&
\frac{w_lc_le^L+w_u c_u e^u}{w_le^L+w_u e^u }\nonumber\\
>&w_l c_l+w_uc_u
\end{align}
\end{samepage}
Note that the inequality is given by the Chebyshev's sum inequality, which states that $w_1, w_2>0$ and $w_1+w_2=1$, one have
$(w_1a_1+w_2a_2)(w_1b_1+w_2b_2)< w_1a_1b_1+w_2a_2b_2$ if $a_1< a_2$ and $b_1< b_2$. This inequality can be easily proved by expanding the left-hand side.
In Eq.~\ref{eq:m0}, we have $w_l>0$, $w_u>0$\,\footnote{$w_u=0$ would lead to $\E^{\Pr{Q}_{\bar{Z}}}(C|\SgAlg{F}^0_t)(\omega)=\E^{\Pr{Q}_{\bar{Z}}}(C\ind{C<u}|\SgAlg{F}^0_t)(\omega)<u(\omega)$ in contradiction with the definition of $u$.} and
\begin{align*}
c_l=\frac{\E^{\Pr{Q}_{\bar{Z}}}\left(\left.Ce^C\ind{C< u}\right|\SgAlg{F}^0_t\right)(\omega)}{\E^{\Pr{Q}_{\bar{Z}}}\left(\left.e^C\ind{C< u}\right|\SgAlg{F}^0_t\right)(\omega)}
<\frac{\E^{\Pr{Q}_{\bar{Z}}}\left(\left.ue^C\ind{C< u}\right|\SgAlg{F}^0_t\right)(\omega)}{\E^{\Pr{Q}_{\bar{Z}}}\left(\left.e^C\ind{C< u}\right|\SgAlg{F}^0_t\right)(\omega)}=u\leqslant
c_u
\end{align*}
and 
\begin{align*}
e^L=\frac{\E^{\Pr{Q}_{\bar{Z}}}\left(\left.e^C\ind{C<u}\right|\SgAlg{F}^0_t\right)(\omega)}{\E^{\Pr{Q}_{\bar{Z}}}\left(\left.\ind{C<u}\right|\SgAlg{F}^0_t\right)(\omega)}
<\frac{\E^{\Pr{Q}_{\bar{Z}}}\left(\left.e^u\ind{C<u}\right|\SgAlg{F}^0_t\right)(\omega)}{\E^{\Pr{Q}_{\bar{Z}}}\left(\left.\ind{C<u}\right|\SgAlg{F}^0_t\right)(\omega)}=e^u
\end{align*}
Therefore Chebyshev's sum inequality is applicable.

We further apply Jensen's inequality to the following expression twice ($x\ln x$ is a convex function while $\ln x$ is a concave function),
\begin{align}\label{eq:jensen}
\E(x\ln x)\geqslant \E(x)\ln \E(x)\geqslant \E(x)\E(\ln x)
\end{align}
Following the inequality above, we take expectation w.r.t $\SgAlg{F}^0_t$ and under the alternative measure generated by the Radon-Nikodym derivative 
\begin{equation*}
\frac{{\bar{Z}}(T)}{{\bar{Z}}(t)}\frac{\ind{C<u}}{\E^{\Pr{Q}_{\bar{Z}}}\left(\ind{C<u}\,|\,\SgAlg{F}^0_t\right)}
\end{equation*}
By further assigning $x=e^C$, we get the following inequality
\begin{align*}
&\frac{\E\left(\left.{\bar{Z}}(T)Ce^C \ind{C<u}\,\right|\SgAlg{F}^0_t\right)(\omega)}{\E\left(\left.{\bar{Z}}(T)\ind{C<u}\,\right|\SgAlg{F}^0_t\right)(\omega)}\geqslant \frac{\E\left(\left. {\bar{Z}}(T)e^C\ind{C<u}\,\right|\SgAlg{F}^0_t\right)(\omega)}{\E\left(\left.{\bar{Z}}(T)\ind{C<u}\,\right|\SgAlg{F}^0_t\right)(\omega)} \times \frac{\E\left(\left. {\bar{Z}}(T)C\ind{C<u}\,\right|\SgAlg{F}^0_t\right)(\omega)}{\E\left(\left.{\bar{Z}}(T)\ind{C<u}\,\right|\SgAlg{F}^0_t\right)(\omega)}\\
\Rightarrow\,&
\frac{\E^{\Pr{Q}_{\bar{Z}}}\left(\left.Ce^C \ind{C<u}\,\right|\SgAlg{F}^0_t\right)(\omega)}{\E^{\Pr{Q}_{\bar{Z}}}\left(\left.e^C \ind{C<u}\,\right|\SgAlg{F}^0_t\right)(\omega)}\geqslant \frac{\E^{\Pr{Q}_{\bar{Z}}}\left(\left.C\ind{C<u}\,\right|\SgAlg{F}^0_t\right)(\omega)}{\E^{\Pr{Q}_{\bar{Z}}}\left(\left.\ind{C<u}\,\right|\SgAlg{F}^0_t\right)(\omega)}
\end{align*}
The LHS is simply $c_l$. Substituting the inequality into Eq.~\ref{eq:m0} one gets
\begin{align*}
\E^{\Pr{Q}_{Z'}}\left(\left.C\,\right|\SgAlg{F}^0_t\right)(\omega)>
\E^{\Pr{Q}_{Z'}}\left(\left.C\ind{C<u}\,\right|\SgAlg{F}^0_t\right)(\omega)+
\E^{\Pr{Q}_{\bar{Z}}}\left(\left.C\ind{C\geq u}\,\right|\SgAlg{F}^0_t\right)(\omega)
=
\E^{\Pr{Q}_{\bar{Z}}}\left(\left.C\,\right|\SgAlg{F}^0_t\right)(\omega)
\end{align*}
This violates the condition stated in Eq.~\ref{eq:le21}. We therefore conclude that $\bigl(C(t,\cdot)\bigr)_{t\in[0,T]}$ admits a progressively measurable modification $\bigl(\tilde{C}(t,\cdot)\bigr)_{t\in[0,T]}$. 
\end{proof}
A process $C$ that satisfies the conditions in Lemma \ref{le:1} admits a progressively measurable modification $\tilde{C}$ w.r.t $\Pr{Q}_{\bar{Z}}$, but not necessarily w.r.t the reference measure $\Pr{P}$. However, if it also holds w.r.t $\Pr{P}$, then we get the converse of Lemma~\ref{le:1}. In fact, for $\bar{Z}\in\M_+(1)$ and any $Z\in\mathcal{Z}(t,\bar{Z})$, both $\Pr{Q}_{\bar{Z}}$ and $\Pr{Q}_{Z}$ are absolutely continuous w.r.t $\Pr{P}$, implying $\tilde{C}$ is a modification of $C$ w.r.t $\Pr{Q}_{\bar{Z}}$ and $\Pr{Q}_{Z}$. This results in
\begin{equation*}
\E^{\Pr{Q}_{\bar{Z}}}\left(\left.C(t,\,\cdot\,)\,\right|\SgAlg{F}^0_t\right) 
= \E^{\Pr{Q}_{\bar{Z}}}\bigl(\tilde{C}(t,\,\cdot\,)\,|\SgAlg{F}^0_t\bigr) 
\quad\text{and}\quad
\E^{\Pr{Q}_{Z}}\left(\left.C(t,\,\cdot\,)\,\right|\SgAlg{F}^0_t\right)=
\E^{\Pr{Q}_{Z}}\bigl(\tilde{C}(t,\,\cdot\,)\,|\SgAlg{F}^0_t\bigr)
\end{equation*}
The progressively measurable process $\tilde{C}$ is adapted to the filtration $\Fltrn{F}^0$. Therefore
\begin{align*}
\E^{\Pr{Q}_{\bar{Z}}}\left(\left.C(t,\,\cdot\,)\,\right|\SgAlg{F}^0_t\right) 
=\tilde{C}(t,\,\cdot\,)=
\E^{\Pr{Q}_{Z}}\left(\left.C(t,\,\cdot\,)\,\right|\SgAlg{F}^0_t\right)
\end{align*}
for all $t\in[0,T]$ and $Z\in\mathcal{Z}(t,{\bar{Z}})$. 
We use Lemma~\ref{le:1} to prove the following proposition.
\begin{proposition}\label{pro:2}
$Z^*\in\M_+(1)$ is a worst-case martingale density iff the random variable 
\begin{align*}
\Omega\ni \omega\mapsto\ell(T,\omega)-\frac{f'(Z^*(T)(\omega))}{\vartheta}
\end{align*}
equals constant $\Pr{Q}_{Z^*}$-a.s., and is dominated by the same constant $\Pr{P}$-a.s.
\end{proposition}
\begin{proof}
Suppose ${Z^*}\in\M_+(1)$ is a worst-case martingale density. According to Definition.~\ref{def:z*}, 
\begin{equation}\label{eq:maxsolve}
{Z^*}=\arg\max_{Z\in\mathcal{Z}(t,{Z^*})}\E^{\Pr{Q}_{Z}}\bigl(\widehat{\ell}(T,Z)-\widehat{\ell}(t,Z)\,\bigl|\,\SgAlg{F}^0_t\bigr)
\end{equation}
for all $t\in[0,T]$. 
Given any $t\in[0,T]$ and any $Z\in\mathcal{Z}(t,{Z^*})$, we construct a new martingale density that lies between ${Z^*}$ and $Z$ by
\begin{align*}
Z_\lambda=(1-\lambda){Z^*}+\lambda Z
\end{align*}
where $\lambda\in[0,1]$. 
$Z_\lambda\in\mathcal{Z}(t,{Z^*})$ for all $\lambda\in[0,1]$ due to the convexity of $\mathcal{Z}(t,{Z^*})$. 
Since ${Z^*}$ solves Eq.~\ref{eq:maxsolve}, the maximum value of 
\begin{align}\label{eq:K0}
K(\lambda)\coloneqq\E^{\Pr{Q}_{Z_\lambda}}\bigl(\widehat{\ell}(T,Z_\lambda)-\widehat{\ell}(t,Z_\lambda)\,\bigl|\,\SgAlg{F}^0_t\bigr)
\end{align}
is reached when $\lambda=0$. 
Taking the first and second derivatives with respect to $\lambda$, we get
\begin{align}\label{eq:Kl}
K'(\lambda)=&\,
\E\biggl(\left.\frac{Z(T)-{Z^*}(T)}{{Z^*}(t)}\left(\ell(T,{Z^*})-\ell(t,\cdot)-\frac{f'\left(Z_\lambda(T)\right)}{\vartheta}\right)\,\right|\SgAlg{F}^0_t\biggr)\\
K''(\lambda)=&-
\E\left(\left.\frac{\left(Z(T)-{Z^*}(T)\right)^2}{{Z^*}(t)}\frac{f''\left(Z_\lambda(T)\right)}{\vartheta}\,\right|\SgAlg{F}^0_t\right)\label{eq:Kl1}
\end{align}
Notice that the twice-differentiable function $f:\mathbb{R}_+\to\mathbb{R}$ is convex as required by the non-negativity of the $f$-divergence \cite[]{ali1966general}. This implies that $f''(z)>0$ for all $z\in\mathbb{R}_+$. Combined with Eq.~\ref{eq:Kl1}, this condition leads to $K''(\lambda)<0$ for all $\lambda\in[0,1]$. For $K(0)=\max_{\lambda\in[0,1]}K(\lambda)$ to hold, the first derivative at $\lambda=0$ must satisfy
\begin{align}\label{eq:ineq}
K'(0)\leqslant 0\Leftrightarrow \E^{\Pr{Q}_{Z}}\left(\left.C_{Z^*}(t,\,\cdot\,)\,\right|\SgAlg{F}^0_t\right)\leqslant \E^{\Pr{Q}_{Z^*}}\left(\left.C_{Z^*}(t,\,\cdot\,)\,\right|\SgAlg{F}^0_t\right) 
\end{align}
where the process $C_{Z^*}:[0,T]\times\Omega\to\mathbb{R}$ is defined by
\begin{align*}
C_{Z^*}(t,\,\cdot\,)\coloneqq\ell(T,\,\cdot\,)-\ell(t,\,\cdot\,)-\frac{f'({Z^*}(T))}{\vartheta}
\end{align*}
The inequality above holds for all $t\in[0,T]$ and all $Z\in\mathcal{Z}(t,{Z^*})$. 
According to Lemma.~\ref{le:1}, $C_{Z^*}$ admits a progressively measurable modification, say $\tilde{C}_{Z^*}$. In particular, at $t=0$
\begin{align*}
C_{Z^*}(0,\,\cdot\,)=\ell(T,\,\cdot\,)-\frac{f'({Z^*}(T))}{\vartheta}
\end{align*}
takes a constant value $c:=\tilde{C}_{Z^*}(0,0)$, $\Pr{Q}_{Z^*}$-a.s. In fact, $\tilde{C}_{Z^*}$ is regarded as a non-anticipative functional so that $\tilde{C}_{Z^*}(0,\omega)=\tilde{C}_{Z^*}(0,0)=c$ for all $\omega\in\Omega$ satisfying $(0,\omega)\sim(0,0)$. As a result, 
\begin{samepage}
\begin{align}\label{eq:QZ1}
\Pr{Q}_{Z^*}\bigl(C_{Z^*}(0,\cdot)=c\bigr)
\geqslant&\,
\Pr{Q}_{Z^*}\bigl(\tilde{C}_{Z^*}(0,\cdot)= c\bigr)
-\Pr{Q}_{Z^*}\bigl(C_{Z^*}(0,\cdot)\neq\tilde{C}_{Z^*}(0,\cdot)\bigr)\nonumber\\
=&\,\Pr{Q}_{Z^*}\bigl(\tilde{C}_{Z^*}(0,\cdot)= c\bigr)-0
\geqslant
\Pr{Q}_{Z^*}\bigl((0,\cdot)\sim(0,0)\bigr)\nonumber\\
=&\,\Pr{Q}_{Z^*}(\{\omega\in\Omega\,|\,\omega(0)=0\})=1
\end{align}
\end{samepage}
Next we prove $\Pr{P}\bigl(C_{Z^*}(0,\cdot)\leqslant c\bigr)=1$ by contradiction. Suppose on the contrary that $\Pr{P}\bigl(C_{Z^*}(0,\cdot)>c\bigr)>0$. We construct a martingale density $Z'\in\mathcal{Z}(0,{Z^*})=\M_+(1)$ by setting
\begin{align*}
Z'(t)=\frac{\E\left(\left.\ind{C_{Z^*}(0,\cdot)>c}\,\right|\SgAlg{F}^0_t\right)}{\Pr{P}\bigl(C_{Z^*}(0,\cdot)>c\bigr)}
\end{align*}
for all $t\in[0,T]$. This leads to
\begin{align*}
\E^{\Pr{Q}_{Z'}}\left(\left.C_{Z^*}(0,\,\cdot\,)\,\right|\SgAlg{F}^0_0\right)
=
\E\left(Z'(T)C_{Z^*}(0,\,\cdot\,)\right)
=\frac{\E\left(\ind{C_{Z^*}(0,\cdot)>c}C_{Z^*}(0,\,\cdot\,)\right)}{\Pr{P}\bigl(C_{Z^*}(0,\cdot)>c\bigr)}>\frac{c\E\left(\ind{C_{Z^*}(0,\cdot)>c}\right)}{\Pr{P}\bigl(C_{Z^*}(0,\cdot)>c\bigr)}=c
\end{align*}
Because we have already shown that $C_{Z^*}(0,\,\cdot\,)=c$, $\Pr{Q}_{Z^*}$-a.s. (Eq.~\ref{eq:QZ1}),
\begin{align*}
\E^{\Pr{Q}_{Z^*}}\left(\left.C_{Z^*}(0,\,\cdot\,)\,\right|\SgAlg{F}^0_0\right) 
= c<\E^{\Pr{Q}_{Z'}}\left(\left.C_{Z^*}(0,\,\cdot\,)\,\right|\SgAlg{F}^0_0\right)
\end{align*}
According to Eq.~\ref{eq:ineq}, $K'(0)> 0$ (where the generic density process $Z$ is replaced by the constructed process $Z'\in\mathcal{Z}(0,Z^*)$). This contradicts the assumption that $Z^*$ is a worst-case martingale density. 

Conversely, given a process $Z^*\in\M_+(1)$, suppose $C_{Z^*}(0,\,\cdot\,):\Omega\to\mathbb{R}$ takes a constant value, say $c$, $\Pr{Q}_{Z^*}$-a.s., and $C_{Z^*}(0,\,\cdot\,)\leqslant c$ $\Pr{P}$-a.s. Given any $t\in[0,T]$ and any $Z\in\mathcal{Z}(t,{Z^*})$, $C_{Z^*}(0,\,\cdot\,)\leqslant c$ $\Pr{Q}_{Z}$-a.s. due to the absolute continuity of $\Pr{Q}_{Z}$ w.r.t. $\Pr{P}$. These properties lead to conditional expectations 
\begin{equation*}
\E^{\Pr{Q}_{Z^*}}\left(\left.C_{Z^*}(0,\,\cdot\,)\,\right|\SgAlg{F}^0_t\right) 
= c
\quad\text{and}\quad
\E^{\Pr{Q}_{Z}}\left(\left.C_{Z^*}(0,\,\cdot\,)\,\right|\SgAlg{F}^0_t\right)\leqslant
c
\end{equation*}
Noticing that $C_{Z^*}(t,\,\cdot\,)=C_{Z^*}(0,\,\cdot\,)-\ell(t,\cdot)$ where $\ell(t,\cdot)$ is $\SgAlg{F}^0_t$-measurable, We have
\begin{equation*}
 \E^{\Pr{Q}_{Z}}\left(\left.C{Z^*}(t,\,\cdot\,)\,\right|\SgAlg{F}^0_t\right)
 \leqslant c-\ell(t,\,\cdot\,)=
\E^{\Pr{Q}_{Z^*}}\left(\left.C_{Z^*}(t,\,\cdot\,)\,\right|\SgAlg{F}^0_t\right)
\end{equation*}
According to Eq.~\ref{eq:ineq}, $K'(0)\leqslant0$. Because $K''(\lambda)<0$ (Eq.~\ref{eq:Kl1}) for all $\lambda\in[0,1]$, $K(0)\geqslant K(1)$. According to the definition of $K(\lambda)$ (Eq.~\ref{eq:K0}), we have
\begin{equation*}
\E^{\Pr{Q}_{Z^*}}\bigl(\widehat{\ell}(T,{Z^*})-\widehat{\ell}(t,{Z^*})\,\bigl|\,\SgAlg{F}^0_t\bigr)=K(0)\geqslant K(1)=\E^{\Pr{Q}_{Z}}\bigl(\widehat{\ell}(T,Z)-\widehat{\ell}(t,Z)\,\bigl|\,\SgAlg{F}^0_t\bigr)
\end{equation*}
This inequality applies to every $t\in[0,T]$ and every $Z\in\mathcal{Z}(t,{Z^*})$. As a result, ${Z^*}$ solves Eq.~\ref{eq:maxsolve} for all $t\in[0,T]$ and is indeed a worst-case martingale density. 
\end{proof}

It is noted that Proposition.~\ref{pro:1} is a general result that works for any $\Fltrn{F}^0$-adapted process $(\widehat{\ell}(t,Z))_{t\in[0,T]}$, irrespective of its actual formulation (Eq.~\ref{eqSec3:defhatl}). On the other hand, Proposition.~\ref{pro:2} makes use of the formulation, thus specifying the condition of a worst-case martingale density w.r.t the function $f(x)$. 
Note that any worst-case density process ${Z^*}\in\M_+(1)$ solves the original problem formulated in Eq.~\ref{eqSec3:IntProb}. Assuming the existence of such ${Z^*}$, we regard Eq.~\ref{eqSec3:IntProb} as the initial value (at $t=0$) of a particular process, termed as the value process. In general, we define three $\Fltrn{F}^0$-adapted processes as below.

\begin{definition}\label{def:4}
Given $\vartheta\in(0,\infty)$ and a worst-case martingale density ${Z^*}\in\M_+(1)$, 
the value process, $U:[0,T]\times\Omega\to\mathbb{R}$, the worst-case risk, $V:[0,T]\times\Omega\to\mathbb{R}$, and the budget process $\eta:[0,T]\times\Omega\to\mathbb{R}$,\footnote{We name it the budget process as it measures the remaining budget of the fictitious adversary \cite[]{glasserman2014robust}. $\eta(0,\cdot)$ is referred as the relative entropy budget in \cite{glasserman2014robust}.} regarded as non-anticipative functionals, are defined by
\begin{samepage}
\begin{align*}
U(t,\cdot)\coloneqq&\widehat{L}(t,{Z^*})+\ell(t,\cdot)\\
V(t,\cdot)\coloneqq&\widehat{L}(t,{Z^*})+\widehat\ell(t,{Z^*})+F(t,{Z^*})\\
\eta(t,\cdot)\coloneqq&\vartheta\left(V(t,\cdot)-U(t,\cdot)\right)
\end{align*}
\end{samepage}
where $\left(F(t,{Z^*})\right)_{t\in[0,T]}$ is the $\Pr{Q}_{Z^*}$-martingale that satisfies $F(T,{Z^*})=f({Z^*}(T))/ {Z^*}(T)$.
\end{definition}
Intuitively, $U(t,\cdot)$ gives the worst-case expected loss, subtracting the on-going cost of perturbing the nominal model from time $t$ to $T$. According to the definition of the worst-case martingale density (Eq.~\ref{eq:def0}),
\begin{align}\label{eq:Utdef}
U(t,\cdot)=&\,\E^{\Pr{Q}_{Z^*}}\bigl(\widehat{\ell}(T,{Z^*})-\widehat{\ell}(t,{Z^*})\,\bigl|\,\SgAlg{F}^0_t\bigr)+\ell(t,\cdot)\nonumber\\
=&\,\E^{\Pr{Q}_{Z^*}}\bigl({\ell}(T,\cdot)\,\bigl|\,\SgAlg{F}^0_t\bigr)-
\vartheta^{-1} {Z^*}(t)^{-1}
{\E\bigl(f\left({Z^*}(T)\right)-f\left({Z^*}(t)\right)\bigl|\,\SgAlg{F}^0_t\bigr)} \ind{{Z^*}(t)>0}
\end{align}
The second term is the penalization term for perturbing the nominal model from time $t$ onwards. For continuity it is defined to be zero in the limiting case of ${Z^*}(t)=0$. 
According to Definition~\ref{def:4},
$V(t,\cdot)$ is the worst-case expected loss, 
\begin{align*}
V(t,\cdot)=\E^{\Pr{Q}_{Z^*}}\bigl(\widehat{\ell}(T,{Z^*})\,\bigl|\,\SgAlg{F}^0_t\bigr)
+\vartheta^{-1}{Z^*}(t)^{-1}\E^{\Pr{Q}}\bigl(f({Z^*}(T))\,\bigl|\,\SgAlg{F}^0_t\bigr)\ind{{Z^*}(t)>0}
=\E^{\Pr{Q}_{Z^*}}\bigl({\ell}(T,\cdot)\,\bigl|\,\SgAlg{F}^0_t\bigr)
\end{align*}
The difference between $V(t,\cdot)$ and $U(t,\cdot)$ gives the cost for perturbing the nominal model (measured by the $f$-divergence), characterized by the process $\eta$:
\begin{align*}
\eta(t,\cdot)= {Z^*}(t)^{-1}
{\E\bigl(f\left({Z^*}(T)\right)-f\left({Z^*}(t)\right)\bigl|\,\SgAlg{F}^0_t\bigr)}\ind{{Z^*}(t)>0}
\end{align*}

We may further consider the terminal and initial values of the three processes. 
The value process, $U(t,\cdot)$, measures the target formulated in Eq.~\ref{eqSec3:IntProb} from backwards, in the sense that
\begin{align}\label{eq:UT}
U(T)={\ell}(T,\cdot)
\quad\text{and}\quad
U(0)=\E^{\Pr{Q}_{Z^*}}\left(\widehat{\ell}(T,{Z^*})\right)
=\max_{Z\in\M_+(1)}\E^{\Pr{Q}_Z}\left(\widehat{\ell}(T,Z)\right)
\end{align}
The worst-case risk process measures the model risk, Eq.~\ref{eqSec2:ModelRisk}, from backwards. According to Lemma.~\ref{le:0}(4), the worst-case density $Z^*$ solves the primal problem with $\eta:=\eta(0,\cdot)=\E\left(f({Z^*}(T))\right)$. Therefore
\begin{align*}
V(T)={\ell}(T,\cdot)
\quad\text{and}\quad
V(0)=\E^{\Pr{Q}_{Z^*}}\left(\ell(T,\cdot)\right)=
\sup_{Z\in\mathcal{Z}_\eta}\E^{\Pr{Q}_{Z^*}}\bigl(\ell(T,\,\cdot\,)\bigr)
\end{align*}
The cumulative budget $\eta$ (i.e. relative entropy budget in \cite{glasserman2014robust}) is measured by the budget process from backwards,
\begin{align*}
\eta(T)=0
\quad\text{and}\quad
\eta(0)=\E\left(f({Z^*}(T))\right)=\eta
\end{align*} 
To solve the problem formulated in Eq.~\ref{eqSec3:IntProb}, Eq.~\ref{eq:UT} suggests solving the process $U$ by backward induction. In a similar way, the model risk, Eq.~\ref{eqSec2:ModelRisk}, and its corresponding cumulative budget, $\eta$, may be quantified by solving the processes $V$ and $\eta$ by backward induction. The full procedure is given by the following theorem. 
\begin{theorem}\label{th:0}
Given $\vartheta\in(0,\infty)$, suppose there exists a function $z:\mathbb{R}\to \mathbb{R}_+$ that satisfies
\begin{align}\label{eq:zx}
x-\frac{f'(z(x))}{\vartheta}=c&\qquad\text{if}\quad x\in I_c\coloneqq\{\vartheta^{-1}y+c\,|\,y\in \text{range}(f')\}\\
z(x)=0&\qquad\text{if}\quad x\notin I_c
\end{align}
where $c\in\mathbb{R}$ is a constant such that $\E\bigl(z\circ\ell(T,\cdot)\bigr)=1$ and $\Pr{P}\left(\ell(T,\cdot)<\sup I_c\right) =1$. Then the value process, $U$, the worst-case risk, $V$, and the budget process, $\eta$, satisfy the following equations
\begin{align}\label{eq:uveta}
U(t,\cdot)=&\,\frac{M(t)+f(Z(t))}{\vartheta Z(t)}+c\nonumber\\
V(t,\cdot)=&\,\frac{W(t)}{Z(t)}\\
\eta(t,\cdot)=&\frac{\vartheta W(t)- M(t)-f(Z(t))}{Z(t)}-\vartheta c\nonumber
\end{align}
for all $t\in[0,T]$ and a.a. $\omega\in\{Z(t)>0\}$, where $(Z,\,M,\,W)$ is a $\Fltrn{F}^0$-adapted $\Pr{P}$-martingale that satisfies the following terminal condition:
\begin{align}\label{eq:vec}
\begin{pmatrix}
Z\\
M\\W
\end{pmatrix}(T)=
\begin{pmatrix}
z\circ\ell(T,\cdot)\\
f'\circ z\circ\ell(T,\cdot)\times {z\circ\ell(T,\cdot)}-{f\circ z\circ\ell(T,\cdot)}
\\
z\circ\ell(T,\cdot)\times\ell(T,\cdot)
\end{pmatrix}
\end{align}
\end{theorem}

\begin{proof}
The function $z$ defined by Eq.~\ref{eq:zx} provides a martingale density $Z\in\M_+(1)$ by composition:
\begin{align}\label{eq:Zdef}
Z(t)=\E\left(\left.z\circ\ell(T,\cdot)\,\right|\SgAlg{F}^0_t\right)
\end{align}
for all $t\in[0,T]$. $Z$ is exactly the first element of the vectorized process defined in Eq.~\ref{eq:vec}. 
It is indeed an element of $\M_+(1)$, for $Z(T)=z\circ\ell(T,\cdot)\geqslant0$ and $Z(0)=\E\bigl(z\circ\ell(T,\cdot)\bigr)=1$. 
The random variable 
\begin{align}\label{eq:compos}
C_Z(0,\cdot)\coloneqq
\ell(T,\cdot)-\frac{f'(Z(T,\cdot))}{\vartheta}=\ell(T,\cdot)-\frac{f'\circ z\circ\ell(T,\cdot)}{\vartheta}
\end{align}
is equal to the constant $c$ $\Pr{Q}_Z$-a.s. In fact, $c\in\mathbb{R}$ is selected such that
\begin{align*}
1=\E\bigl(z\circ\ell(T,\cdot)\bigr)
=\E\left(z\circ\ell(T,\cdot)\ind{\ell(T,\cdot)\in I_c}\right)+\E\left(z\circ\ell(T,\cdot)\ind{\ell(T,\cdot)\notin I_c}\right)
=\E\left(z\circ\ell(T,\cdot)\ind{\ell(T,\cdot)\in I_c}\right)
\end{align*}
by virtue of $z(x)=0$ for all $x\notin I_c$. Since $C_Z(0,\omega)=c$ for all $\omega\in\Omega$ satisfying $\ell(T,\omega)\in I_c$, we have
\begin{align*}
\Pr{Q}_Z\bigl(C_Z(0,\cdot)=c\bigr)=
\E\left(Z(T)\ind{C_Z(0,\cdot)=c}\right)
\geqslant&\,\E\left(z\circ\ell(T,\cdot)\ind{\ell(T,\cdot)\in I_c}\ind{C_Z(0,\cdot)=c}\right)\\
=&\,\E\left(z\circ\ell(T,\cdot)\ind{\ell(T,\cdot)\in I_c}\right)=1
\end{align*}

Next we need to show that $C_Z(0,\cdot)\leqslant c$ $\Pr{P}$-a.s. 
Notice that the function $f':(0,\infty)\to\mathbb{R}$ is continuous and strictly increasing due to the convexity of $f$, implying that $\text{range}(f')=(f'(0_+),f'(\infty_-))$. 
We conclude that $ \text{range}(f')$ is an open interval and denote it by $(a,b)$, where $a$ and $b$ can be either real numbers or $\pm\infty$. 
According to the assumption, we have
\begin{align*}
1=\Pr{P}\left(\ell(T,\cdot)<\sup I_c\right) 
=\Pr{P}\left(\ell(T,\cdot)\in I_c\bigcup\ell(T,\cdot)\leqslant \{\vartheta^{-1}a+c\}\right)
\end{align*}
We extend the function $f'$ continuously to zero by assigning $f'(0)=a$. 
\begin{align*}
\Pr{P}\bigl(C_Z(0,\cdot)\leqslant c\bigr)
=&\,\E\left(\ind{\ell(T,\cdot)\in I_c}\ind{C_Z(0,\cdot)\leqslant c}\right)+\E\left(\ind{\ell(T,\cdot)\notin I_c}\ind{C_Z(0,\cdot)\leqslant c}\right)\\
=&\,E\left(\ind{\ell(T,\cdot)\in I_c}\right)+\E\left(\ind{\ell(T,\cdot)\notin I_c}\ind{\ell(T,\cdot)-\vartheta^{-1}f'(0)\leqslant c}\right)\\
=&\,\E\left(\ind{\ell(T,\cdot)\in I_c}\right)+\E\left(\ind{\ell(T,\cdot)\notin (\vartheta^{-1}a+c,\,\vartheta^{-1}b+c)}\ind{\ell(T,\cdot)\leqslant  \vartheta^{-1}a+c}\right)\\
=&\,\Pr{P}\left(\ell(T,\cdot)\in I_c\bigcup\ell(T,\cdot)\leqslant \vartheta^{-1}a+c\right)=1
\end{align*}
We conclude that $C_Z(0,\cdot)=c$ $\Pr{Q}_Z$-a.s. and $C_Z(0,\cdot)\leqslant c$ $\Pr{P}$-a.s.  
According to Proposition.~\ref{pro:2}, $Z$ defined in Eq.~\ref{eq:Zdef} is a worst-case density process.

The second component of Eq.~\ref{eq:vec} is a $\Pr{P}$-martingale given by
\begin{align*}
M(t)=&\,\E\left(\left.f'\circ z\circ\ell(T,\cdot)\times {z\circ\ell(T,\cdot)}-{f\circ z\circ\ell(T,\cdot)}\,\right|\SgAlg{F}^0_t\right)\\
=&\,\E\left(\left.Z(T)f'(Z(T))-{f(Z(T))}\,\right|\SgAlg{F}^0_t\right)
\end{align*}
for all $t\in[0,T]$. Substituting Eq.~\ref{eq:compos} into Eq.~\ref{eq:Utdef}, we have
\begin{align*}
U(t,\cdot)
=&\,\E^{\Pr{Q}_{Z}}\left(\left.
C_Z(0,\cdot)+\frac{f'(Z(T,\cdot))}{\vartheta}\,\right|\,\SgAlg{F}^0_t\right)-\E\left(\left.\frac{f(Z(T))-f(Z(t))}{\vartheta Z(t)}\,\right|\SgAlg{F}^0_t\right)\\
=&\,c+\E\left(\left.\frac{Z(T)f'(Z(T))-f(Z(T))+f(Z(t))}{\vartheta  Z(t)}\,\right|\SgAlg{F}^0_t\right)\\
=&\,\frac{M(t)+{f(Z(t))}}{\vartheta Z(t)}+c
\end{align*}
By virtue of $C_Z(0,\cdot)=c$ $\Pr{Q}_Z$-a.s., the equation above holds $\Pr{Q}_Z$-a.s. More precisely, it holds for a.a. $\omega\in\{Z(t)>0\}$.\footnote{\,
According to the definition of $C_Z(0,\cdot)$ (Eq.~\ref{eq:compos}), $C_Z(0,\omega)=c$ for all $\omega\in\Omega$ satisfying $\ell(T,\omega)\in I_c$. It follows from $Z(T)(\omega)=z\circ\ell(T,\omega)=0$ for a.a $\omega\in\{\omega\in\Omega\,|\,\ell(T,\omega)\notin I_c\}$ that
\begin{align*}
\E^{\Pr{Q}_Z}(C_Z(0,\cdot)\,|\,\SgAlg{F}^0_t)=&\,\E^{\Pr{Q}_Z}(C_Z(0,\cdot)\ind{\ell(T,\cdot)\in I_c}\,|\,\SgAlg{F}^0_t)+\E^{\Pr{Q}_Z}(C_Z(0,\cdot)\ind{\ell(T,\cdot)\notin I_c}\,|\,\SgAlg{F}^0_t)\\
=&\,Z(t)^{-1}\E(Z(T,\cdot)C_Z(0,\cdot)\ind{\ell(T,\cdot)\in I_c}\,|\,\SgAlg{F}^0_t)+Z(t)^{-1}\E(Z(T,\cdot)C_Z(0,\cdot)\ind{\ell(T,\cdot)\notin I_c}\,|\,\SgAlg{F}^0_t)\\
=&\,cZ(t)^{-1}\E(Z(T,\cdot)\ind{\ell(T,\cdot)\in I_c}\,|\,\SgAlg{F}^0_t)\\
=&\,cZ(t)^{-1}\E(Z(T,\cdot)\,|\,\SgAlg{F}^0_t)=c
\end{align*}
for a.a. $\omega\in\{Z(t)>0\}$.
}
The third element of Eq.~\ref{eq:vec}, $W(t)=\\ \E\bigl(\left.z\circ\ell(T,\cdot)\times\ell(T,\cdot)\,\right|\SgAlg{F}^0_t\bigr)
=\E\bigl(\left.Z(T)\ell(T,\cdot)\,\right|\SgAlg{F}^0_t\bigr)$, characterizes the worst-case risk by
\begin{align*}
V(t,\cdot)=\E^{\Pr{Q}_Z}\left(\ell(T,\cdot)\,\bigr|\SgAlg{F}^0_t\right)
=\E\left(\left.\frac{Z(T)}{Z(t)}\ell(T,\cdot)\,\right|\SgAlg{F}^0_t\right)
=\frac{W(t)}{Z(t)}
\end{align*}
for all $\omega\in\{\omega\in\Omega\,|\,Z(t)(\omega)>0\}$. Thus the equation above holds $\Pr{Q}_Z$-a.s. 
Following the expressions for $U(t,\cdot)$ and $V(t,\cdot)$, we get the formula for the budget process 
\begin{equation*}
\eta(t,\cdot)=\vartheta\bigl(V(t,\cdot)-U(t,\cdot)\bigr)=\frac{\vartheta W(t)-M(t)-{f(Z(t))}}{Z(t)}-\vartheta c
\end{equation*}
\end{proof}

In the proof above, we propose the inverse of the function $f'$, denoted by $g: \text{range}(f')\to(0,\infty)$. Using this inverse function, we have the following proposition which states that certain integrability conditions guarantee the existence of the solution, given by Theorem ~\ref{th:0}, to the problem of model risk quantification. 
\begin{proposition}\label{pro:45}
Denote $g:(a,b)\to(0,\infty)$ as the inverse function of $f'$. If $f'(\infty_-)=\infty$ and for every $c\in\mathbb{R}$ $g\bigl(\vartheta(\ell(T,\cdot)-c)\bigr)\ind{\ell(T,\cdot)\in I_c}$ is integrable under the reference measure $\Pr{P}$, then the assumptions in Theorem ~\ref{th:0} hold.
\end{proposition}
\begin{proof}
We need to prove the existence of $c\in\mathbb{R}$ and $z:\mathbb{R}\to \mathbb{R}_+$, such that Eq.~\ref{eq:zx} for all $x\in I_c$ and $z(x)=0$ for all $x\notin I_c$,  $\E\bigl(z\circ\ell(T,\cdot)\bigr)=1$ and $\Pr{P}\left(\ell(T,\cdot)<\sup I_c\right) =1$. 

We have shown in the proof of Theorem ~\ref{th:0} that $ \text{range}(f')=(a,b)$. Here $b$ takes $\infty$ as the strictly increasing function $f'$ diverges at infinity. 
For a given $c\in\mathbb{R}$, the implicit equation Eq.~\ref{eq:zx} gives  
\begin{align*}
z(x)=g\bigl(\vartheta(x-c)\bigr)
\end{align*}
for all $x\in I_c=(\vartheta^{-1}a+c,\infty)$. 
For all $x\notin I_c$, $z(x)=0$ which gives
\begin{align*}
\E\bigl(z\circ\ell(T,\cdot)\bigr)
=&\,\E\bigl(z\circ\ell(T,\cdot)\ind{\ell(T,\cdot)>\vartheta^{-1}a+c}\bigr)+\E\bigl(z\circ\ell(T,\cdot)\ind{\ell(T,\cdot)\leqslant\vartheta^{-1}a+c}\bigr)\\
=&\,\E\bigl(g\bigl(\vartheta(\ell(T,\cdot)-c)\bigr)\ind{\ell(T,\cdot)>\vartheta^{-1}a+c}\bigr)
\end{align*}
We would like to show that the function $K:\mathbb{R}\to\mathbb{R}$ defined by
\begin{align}\label{eq:contfunc}
K(c):=\E\left(g\bigl(\vartheta(\ell(T,\cdot)-c)\bigr)\ind{\ell(T,\cdot)>\vartheta^{-1}a+c}\right)
\end{align}
takes value of one for some $c\in\mathbb{R}$. 

First we will show that $K$
is continuous. Fix an arbitrary $c_0\in\mathbb{R}$ and $\varepsilon\in(0,\infty)$. Resulted from the continuity of $g$, the function $y(\cdot,\omega):(-\infty,c_0\,]\to\mathbb{R}$ defined by
\begin{align*}
y(c,\omega):=\left(g\bigl(\vartheta(\ell(T,\omega)-c)\bigr)-g\bigl(\vartheta(\ell(T,\omega)-c_0)\bigr)\right)\ind{\ell(T,\omega)>\vartheta^{-1}a+c_0}
\end{align*}
is continuous for every $\omega\in\Omega$. Therefore, the function $Y:(-\infty,c_0]\to\mathbb{R}$, defined by $Y(c):=\E\left(y(c,\cdot)\right)$, is continuous at $c_0$.\footnote{
It follows from the dominated convergence theorem that $Y$ is continuous at $c_0$. In fact, the sequence, $\{y(c_0-1/n,\cdot)\}_{n=1}^\infty$, of real-valued measurable functions
converges pointwise to $y(c_0,\cdot)$ by virtue of its continuity. The sequence is dominated by $y(c_0-1,\cdot)$ due to the fact that $g$ increases monotonically. $y(c_0-1,\cdot)$ is integrable as
\begin{align*}
\E\bigl(|y(c_0-1,\cdot)|\bigr)\leqslant \E\left(g\bigl(\vartheta(\ell(T,\cdot)-c_0+1\bigr)\ind{\ell(T,\omega)>\vartheta^{-1}a+c_0}\right)
\leqslant \E\bigl(g\bigl(\vartheta(\ell(T,\cdot)-c_0+1\bigr)\ind{\ell(T,\omega)>I_{c_0-1}}\bigr)<\infty
\end{align*}
The dominated convergence theorem guarantees the convergence of the expectation 
\begin{align*}
\lim_{n\to\infty}\E\bigl(y(c_0-1/n,\cdot)\bigr)=\E\bigl(y(c_0,\cdot)\bigr)=0
\end{align*}
This means that given an arbitrary $\varepsilon>0$, there exists $n_0\in\mathbb{N}$ such that $
\bigl|\E\bigl(y(c_0-1/n,\cdot)\bigr)\bigr|<\varepsilon$ 
for all $n\geqslant n_0$. Due to the fact that $g$ increases monotonically, for every $c\in[c_0-1/n_0,c_0]$ we have
\begin{align*}
0\leqslant\E\bigl(y(c,\cdot)\bigr)-\E\bigl(y(c_0,\cdot)\bigr)=\E\bigl(y(c,\cdot)\bigr)\leqslant
\E\bigl(y(c_0-1/n,\cdot)\bigr)<\varepsilon
\end{align*}
This proves that $Y$ is continuous at $c_0$.
}
Its continuity implies the existence of $\delta>0$ such that $|Y(c)|=|Y(c)-Y(c_0)|<\varepsilon/2$ for all $c_0\in\mathbb{R}$ satisfying $c_0-\delta<c\leqslant c_0$. 
Let 
\begin{align*}
\delta_-:=\min\left(\delta, \frac{f'(\varepsilon/2)-a}{\vartheta}\right)
\end{align*}
Then for all $c_0-\delta_-<c\leqslant c_0$ we have 
\begin{align}
0\leqslant K(c)-K(c_0)
=\,&\E\left(g\bigl(\vartheta(\ell(T,\cdot)-c)\bigr)\ind{\ell(T,\cdot)-\vartheta^{-1}a\in(c, c_0]}\right)+Y(c)\nonumber\\
\leqslant\,&\E\left(g\bigl(\vartheta(\vartheta^{-1}a+c_0-c)\bigr)\ind{\ell(T,\cdot)-\vartheta^{-1}a\in(c, c_0]}\right)+Y(c)\nonumber\\
<\,&g(a+\vartheta\delta_-)+\varepsilon/2\nonumber\\
\leqslant\,& g\bigl({f'(\varepsilon/2)}\bigr)+\varepsilon/2\nonumber\\
=\,&\varepsilon\nonumber
\end{align}
We may prove in a similar way that there exists $\delta_+>0$ such that $K(c)-K(c_0)\in(-\varepsilon,0\,]$ for all $c_0<c<c_0+\delta_+$. Combining the two arguments, $|K(c)-K(c_0)|$ is less than $\varepsilon$ for all $c\in\mathbb{R}$ satisfying $|c-c_0|<\min(\delta_+,\delta_-)$. This proves that the function $K$, defined in Eq.~\ref{eq:contfunc}, is continuous.

Next we need to prove that there exist $c_+, c_-\in\mathbb{R}$ such that $K(c_+)\leqslant1$ and $K(c_-)\geqslant1$. In fact, the limit $\lim_{c\to -\infty}\Pr{P}\bigl(\ell(T,\cdot)>\vartheta^{-1}a+c\bigr)=1$ implies the existence of $c\in\mathbb{R}$ such that 
$\Pr{P}\bigl(\ell(T,\cdot)>\vartheta^{-1}a+c\bigr)\geqslant 1/\xi$ for some $\xi>1$. Defining 
\begin{align*}
c_-:=c-\frac{\max\bigl(0, f'(\xi)-a\bigr)}{\vartheta}\leqslant c
\end{align*}
we have
\begin{align*}
K(c_-)\geqslant
\E\bigl(g\bigl(\vartheta(\ell(T,\cdot)-c_-)\bigr)\ind{\ell(T,\cdot)>\vartheta^{-1}a+ c}\bigr)
\geqslant&\,
g\bigl(\vartheta(\vartheta^{-1}a+ c-c_-)\bigr)\E\bigl(\ind{\ell(T,\cdot)>\vartheta^{-1}a+ c}\bigr)\\
\geqslant&\,
g\bigl(\vartheta(\vartheta^{-1}a+ \vartheta^{-1}(f'(\xi)-a)\bigr)\E\bigl(\ind{\ell(T,\cdot)>\vartheta^{-1}a+ c}\bigr)\\
=&\,
\xi\Pr{P}\bigl(\ell(T,\cdot)>\vartheta^{-1}a+c\bigr)\\
\geqslant&\, 1
\end{align*}
On the other hand, the following limit\footnote{\,The convergence is guaranteed by the dominated convergence theorem. See the footnote in the last page.}
\begin{align*}
\lim_{c\to \infty}\E\left( g(\vartheta\ell(T,\cdot))\ind{\ell(T,\cdot)-\vartheta^{-1}a\in(0,c)}\right)
=\E\bigl(g(\vartheta\ell(T,\cdot))\ind{\ell(T,\cdot)>\vartheta^{-1}a}\bigr)<\infty
\end{align*}
implies the existence of $c\in\mathbb{R}$ such that 
\begin{align*}
\E\left( g(\vartheta\ell(T,\cdot))\ind{\ell(T,\cdot)-\vartheta^{-1}a\in(0,c)}\right)
\geqslant
\E\bigl(g(\vartheta\ell(T,\cdot))\ind{\ell(T,\cdot)>\vartheta^{-1}a}\bigr)-1
\end{align*}
Letting $c_+=\max(0,c)$, we have
\begin{align*}
K(c_+)
\leqslant&\,\E\left(g(\vartheta\ell(T,\cdot))\ind{\ell(T,\cdot)>\vartheta^{-1}a+c_+}\right)\\
\leqslant&\,\E\left(g(\vartheta\ell(T,\cdot))\ind{\ell(T,\cdot)>\vartheta^{-1}a+c}\right)\\
=&\,\E\bigl(g(\vartheta\ell(T,\cdot))\ind{\ell(T,\cdot)>\vartheta^{-1}a}\bigr)-
\E\left(g(\vartheta\ell(T,\cdot))\ind{\ell(T,\cdot)-\vartheta^{-1}a\in(0,c)}\right)\\
\leqslant&\, 1
\end{align*}
According to the intermediate value theorem, there exists $c\in\mathbb{R}$ such that the continuous function $K$, defined in Eq.~\ref{eq:contfunc}, takes the value of one. 
\footnote{\,Such $c\in\mathbb{R}$ is also unique by noticing that the function $K$ is strictly decreasing.}

The condition $\Pr{P}\left(\ell(T,\cdot)<\sup I_c\right) =1$ holds irrespective of the actual measure $\Pr{P}$, for 
\begin{align*}
\bigcup_{x\in I_c}\ell(T,\cdot)\leqslant x
=\bigcup_{x> \vartheta^{-1}a+c}\{\omega\in\Omega\,|\,\ell(T,\omega)\leqslant x\}=\{\omega\in\Omega\,|\,\ell(T,\omega)\in\mathbb{R}\}
\end{align*}
has probability one. As a result, the assumptions stated in Theorem \ref{th:0} are valid, which guarantees the existence of the worst-case solution provided by the theorem. 
\end{proof}

We consider a special class of $f$-divergence, including the renowned Kullback-Leibler divergence, of which the function $\mathbb{R}\ni x\mapsto xf'(x)-f(x)$ is linear (or equivalently $x\mapsto xf''(x)$ is constant). This type of $f$-divergence has a particular advantage on applying Theorem.~\ref{th:0}, because the process
\begin{align}\label{eq:simple}
M(t)=&\,\E\left(\left.Z(T)f'(Z(T))-{f(Z(T))}\,\right|\SgAlg{F}^0_t\right)\nonumber\\
=&\,
\E(Z(T)\,|\,\SgAlg{F}^0_t)\times f'\left(\E(Z(T)\,|\,\SgAlg{F}^0_t)\right)-f\left(\E(Z(T)\,|\,\SgAlg{F}^0_t)\right)\nonumber\\
=&\,Z(t)f'(Z(t))-{f(Z(t))}
\end{align}
can be calculated directly from $Z(t)$. Therefore in practice we only need to apply backward induction to the two-dimensional $\Pr{P}$-martingale $(Z(t),W(t))_{t\in[0,T]}$. By substituting Eq.~\ref{eq:simple} into Eq.~\ref{eq:uveta}, we have the following proposition. 
\begin{corollary}\label{cor:46}
Suppose in Theorem \ref{th:0} there exists $d\in(0,\infty)$ such that $xf''(x)=d$ for all $x\in\mathbb{R}_+$. 
Then the value process, $U$, the worst-case risk, $V$, and the budget process, $\eta$, satisfy the following equations 
\begin{align}\label{eq:uveta1}
U(t,\cdot)=&\,\frac{f'(Z(t))}{\vartheta}+c\nonumber\\
V(t,\cdot)=&\,\frac{W(t)}{Z(t)}\\
\eta(t,\cdot)=&\frac{\vartheta W(t)}{Z(t)}-f'(Z(t))-\vartheta c\nonumber
\end{align}
for all $t\in[0,T]$ and all $\omega\in\Omega$ such that $Z(t)(\omega)>0$, where $(Z,\,W)$ is a $\Fltrn{F}^0$-adapted $\Pr{P}$-martingale that satisfies the following terminal condition:
\begin{align*}
\begin{pmatrix}
Z\\W
\end{pmatrix}(T)=
\begin{pmatrix}
z\circ\ell(T,\cdot)\\
z\circ\ell(T,\cdot)\times\ell(T,\cdot)
\end{pmatrix}
\end{align*}
\end{corollary}
Corollary \ref{cor:46} applies to the Kullback-Leibler divergence. In particular, the calculation of the constant $c$ is pretty straightforward. We illustrate this in the following corollary.
\begin{corollary}\label{co:47}
Under the Kullback-Leibler divergence, suppose $\E\left(e^{\vartheta\ell(T,\cdot)}\right)<\infty$. Then there exists an unique solution to the problem of model risk quantification, given by
\begin{align*}
U(t,\cdot)=&\,\frac{\ln \tilde{Z}(t)}{\vartheta}\\
V(t,\cdot)=&\,\frac{\tilde{W}(t)}{\tilde{Z}(t)}\\
\eta(t,\cdot)=&\,\frac{\vartheta \tilde{W}(t)}{\tilde{Z}(t)}-\ln \tilde{Z}(t)
\end{align*}
where $\bigl(\tilde{Z},\,\tilde{W}\bigr)$ is a $\Fltrn{F}^0$-adapted $\Pr{P}$-martingale that satisfies the terminal condition:
\begin{align*}
\begin{pmatrix}
\tilde{Z}\\
\tilde{W}
\end{pmatrix}(T)=
\begin{pmatrix}
\exp\left(\vartheta \ell(T,\cdot)\right)\\ 
\ell(T,\cdot)\exp\left(\vartheta \ell(T,\cdot)\right)
\end{pmatrix}
\end{align*}
\end{corollary}
\begin{proof}
The Kullback-Leibler divergence adopts $f'(x)=(x\ln x)'=\ln x+1$ for all $x\in(0,\infty)$. $f'$ diverges at $\infty$. 
The inverse function $g:\mathbb{R}\to(0,\infty)$ is given by $g(x)=e^{x-1}$. Since $\E\left(e^{\vartheta\ell(T,\cdot)}\right)<\infty$, we have
\begin{align*}
\E\left(\bigl|g\bigl(\vartheta(\ell(T,\cdot)-c)\bigr)\ind{\ell(T,\cdot)\in I_c}\bigr|\right)
=e^{-\vartheta c-1}\E\left(e^{\vartheta\ell(T,\cdot)}\right)<\infty
\end{align*}
for all $c\in\mathbb{R}$. Proposition \ref{pro:45} guarantees the existence of a unique $c\in\mathbb{R}$ and $z:\mathbb{R}\to \mathbb{R}_+$ satisfying $\E\bigl(z\circ\ell(T,\cdot)\bigr)=1$, therefore a unique solution to the problem of model risk quantification. 

More specifically, we calculate the function $z:\mathbb{R}\to\mathbb{R}_+$ from Eq.~\ref{eq:zx}:
\begin{align*}
z(x)=e^{\vartheta(x-c)-1}
\end{align*}
for all $x\in\mathbb{R}$. 
The constant $c\in\mathbb{R}$ is given by
\begin{align*}
1=\E\bigl(z\circ\ell(T,\cdot)\bigr)=\E\left(e^{\vartheta\left(\ell(T,\cdot)-c\right)-1}\right)
\Leftrightarrow
c=\frac{1}{\vartheta}\ln\E\left(e^{\vartheta\ell(T,\cdot)-1}\right)=\frac{\ln\tilde{Z}(0)-1}{\vartheta}
\end{align*}

The corollary defines two $\Pr{P}$-martingales by
\begin{align*}
\tilde{Z}(t)=&\,\E\left(\left.e^{\vartheta\ell(T,\cdot)}\,\right|\SgAlg{F}^0_t\right)\\
\tilde{W}(t)=&\,\E\left(\left.\ell(T,\cdot)e^{\vartheta\ell(T,\cdot)}\,\right|\SgAlg{F}^0_t\right)
\end{align*}
The process $Z$ and $W$ in Corollary \ref{cor:46} are simply normalized versions of $\tilde{Z}$ and $\tilde{W}$,
\begin{align*}
Z(t)=&\,\E\left(\left.z\circ\ell(T,\cdot)\,\right|\SgAlg{F}^0_t\right)=\E\left(\left.e^{\vartheta\left(\ell(T,\cdot)-c\right)-1}\,\right|\SgAlg{F}^0_t\right)
=\frac{\tilde{Z}(t)}{\tilde{Z}(0)}
\\
W(t)=&\,\E\left(\left.z\circ\ell(T,\cdot)\times\ell(T,\cdot)\,\right|\SgAlg{F}^0_t\right)=\E\left(\left.\ell(T,\cdot)e^{\vartheta\left(\ell(T,\cdot)-c\right)-1}\,\right|\SgAlg{F}^0_t\right)
=\frac{\tilde{W}(t)}{\tilde{Z}(0)}
\end{align*}
Substituting the equations above into Eq.~\ref{eq:uveta1}, we have
\begin{align*}
U(t,\cdot)=&\,\frac{\ln(Z(t))+1}{\vartheta }+c
=\frac{\ln {\tilde{Z}(t)}}{\vartheta }\\
V(t,\cdot)=&\,\frac{W(t)}{Z(t)}
=\frac{\tilde{W}(t)}{\tilde{Z}(t)}\\
\eta(t,\cdot)=&\,\frac{\vartheta W(t)}{Z(t)}-\bigl(\ln(Z(t))-1\bigr)-\vartheta c
=\frac{\vartheta \tilde{W}(t)}{\tilde{Z}(t)}-\ln \tilde{Z}(t)
\end{align*}
Note that $Z(T)(\omega)=e^{\vartheta(\ell(T,\omega)-c)-1}>0$ for all $\omega\in\Omega$. $Z(t)=\E(Z(T)\,|\SgAlg{F}^0_t)>0$, implying that the equations above hold for all $t\in[0,T]$ and all $\omega\in\Omega$.
\end{proof}

\section{Model Risk Measurement with Continuous Semimartingales}

The last section provides the general theory on quantifying the model risk. In this section, we focus on the class of continuous semimartingales. 
It has an important property formulated by the functional Ito formula. To introduce the formula we need to briefly review the functional Ito calculus \cite[]{bally2016functional}. 
First we define the horizontal derivative and the vertical derivative of a non-anticipative functional $F:\Lambda_T^d\to\mathbb{R}$. Its horizontal derivative at $(t,\omega)\in\Lambda_T^d$ is defined by the limit
\begin{align*}
\mathcal{D}F(t,\omega):=\lim_{h\to0^+}\frac{F(t+h,\omega)-F(t,\omega)}{h}
\end{align*}
if it exists. 
Intuitively, it describes the rate of change w.r.t time, assuming no change of the state variable from $t$ onwards, and conditional to its history up to $t$ given by the stopped path $\omega_t$. 
On the other hand, the vertical derivative describes the rate of change w.r.t the state variable from $t$ onwards. Formally, the vertical derivative at $(t,\omega)\in\Lambda_T^d$, denoted by $\nabla_\omega F(t,\omega)$, is defined as the gradient of the function $\mathbb{R}^d\ni x \mapsto F\bigl(t,\omega_t+x\ind{[t,T]}\bigr)$ at $0$,
assuming its existence. The horizontal and vertical derivatives of a non-anticipative functional are also non-anticipative functionals.

We define the left-continuous non-anticipative functionals by noticing that the space of stopped paths, $\Lambda_T^d$, is endowed with a metric $d_\infty$. Suppose $F:\Lambda_T^d\to\mathbb{R}$ is a non-anticipative functional. $F$ is left-continuous if for every $(t,\omega)\in\Lambda_T^d$ and $\varepsilon>0$, there exists $\delta>0$ such that $|F(t,\omega)-F(t',\omega')|<\varepsilon$ for all $(t',\omega')\in\Lambda_T^d$ satisfying $t'<t$ and $d_\infty((t,\omega),(t',\omega'))<\delta$. 
We may further impose a boundedness condition to a non-anticipative functional $F$. It states that for any compact $K\subset\mathbb{R}^d$ and $t_0<T$, there exists a $C>0$ such that $|F(t,\omega)|\leqslant C$ for all $t\leqslant t_0$ and $\omega\in\Omega$. Suppose a non-anticipative functional $F$ is horizontally differentiable and vertically twice-differentiable for all $(t,\omega)\in\Lambda_T^d$, and $\mathcal{D}F$, $\nabla_\omega F$ and $\nabla_\omega ^2 F$ satisfy the boundedness condition above. In addition, $F$, $\nabla_\omega F$ and $\nabla_\omega ^2 F$ are left-continuous, and $\mathcal{D}F$ is continuous for all $(t,\omega)\in\Lambda_T^d$. Then we call $F$ a regular functional. 

Suppose the canonical process $X$ on $\Omega$ is a continuous semimartingale and $F:\Lambda_T^d\to\mathbb{R}$ is a regular functional. 
The $\mathbb{R}$-valued process $(Y(t))_{t\in[0,T]}$, defined by $Y(t)=F(t,\cdot)$ for all $t\in[0,T]$, follows the functional Ito formula $\Pr{P}$-a.s.\cite[pp. 190--191]{bally2016functional}
\begin{align*}
Y(t)-Y(0)=\int_0^t\mathcal{D}F(u,\cdot)du
+\int_0^t\nabla_\omega F(u,\cdot)dX(u)
+\frac{1}{2}\int_0^t\Tr\left(\nabla_\omega^2F(u,\cdot)d[X](u)\right)
\end{align*}
If we further impose the constraint that $\int_0^T\xi(t)dX(t)=0$ for all bounded predictable processes $\xi$ satisfying $\int_0^T\xi(t)dt=0$, then the canonical process $X$ is a strong solution to the SDE \cite[]{revuz2013continuous}
\begin{align}\label{eq:sde}
dX(t)=\mu(t)dt+\sigma(t)dW(t)
\end{align}
where $(W(t))_{t\in[0,T]}$ is a $\mathbb{R}^d$-valued standard Wiener process on the underlying filtered probability space (assuming its existence). 
$(\mu(t))_{t\in[0,T]}$ is a $\mathbb{R}^d$-valued predictable process, and 
$(\sigma(t))_{t\in[0,T]}$ is a $\mathbb{R}^{d^2}$-valued predictable process. We may identify their elements, say $(\mu_i(t))_{t\in[0,T]}$ and $(\sigma_{ij}(t))_{t\in[0,T]}$, with non-anticipative functionals. 
The SDE Eq.~\ref{eq:sde} may be regarded as a path-dependent generalisation of the renowned Ito diffusion process. The existence and uniqueness of its solutions have been given in the literature by imposing various conditions (e.g. boundedness and Lipschitz properties, see \citet{bally2016functional}). 
Now if $X$ satisfies Eq.~\ref{eq:sde} $\Pr{P}$-a.s., then it follows from the functional Ito formula that the process $Y$ is a strong solution to the SDE
\begin{align*}
dY(t)=\left(\mathcal{D}F(t,\cdot)+\mu(t)\nabla_\omega F(t,\cdot)+\frac{\Tr\bigl(\sigma(t)^2\nabla_\omega^2F(t,\cdot)\bigr)}{2}\right)dt
+\sigma(t)\nabla_\omega F(t,\cdot)dW(t)
\end{align*}
Note that the square of $\sigma(t)$ is in the sense of matrix multiplication, i.e. $\sigma(t)^2=\sigma(t)\sigma(t)^T$. 
For simplicity we may define a nonlinear differential operator $\mathcal{A}$ that sends a regular functional to a non-anticipative functional by
\begin{align}\label{eq:Adef}
\mathcal{A}F:=\mathcal{D}F+\mu(t)\nabla_\omega F+\frac{1}{2}\Tr\bigl(\sigma(t)^2\nabla_\omega^2F\bigr)
\end{align}
Then the process $Y$, defined by $Y(t)=F(t,\cdot)$, is a strong solution to
\begin{align}\label{eq:sde1}
dY(t)=\mathcal{A}F(t,\cdot)dt
+\sigma(t)\nabla_\omega F(t,\cdot)dW(t)
\end{align}
Suppose $Y$ is a $\Pr{P}$-martingale, then the regular functional $F$ satisfies $\mathcal{A}F=0$ $\Pr{P}$-a.s. Applying this property, we may convert the martingale statement in Theorem \ref{th:0} to an analytical statement. This is formulated in the following corollary.

\begin{corollary}\label{co:51}
Given $\vartheta\in(0,\infty)$, suppose there exist $c\in\mathbb{R}$ and $z:\mathbb{R}\to \mathbb{R}_+$ defined in Theorem \ref{th:0}. 
If the canonical process $X$ satisfies Eq.~\ref{eq:sde} for some $\mathbb{R}^d$-valued predictable process $(\mu(t))_{t\in[0,T]}$ and $\mathbb{R}^{d^2}$-valued predictable process 
$(\sigma(t))_{t\in[0,T]}$, then the value process, $U$, the worst-case risk, $V$, and the cost process, $\eta$, satisfy the following equations
\begin{align*}
U(t,\cdot)=&\,\frac{M(t)+f(Z(t))}{\vartheta Z(t)}+c\nonumber\\
V(t,\cdot)=&\,\frac{W(t)}{Z(t)}\\
\eta(t,\cdot)=&\,\frac{\vartheta W(t)- M(t)-f(Z(t))}{Z(t)}-\vartheta c\nonumber
\end{align*}
for all $t\in[0,T]$ and all $\omega\in\Omega$ such that $Z(t)(\omega)>0$, 
where $Z$, $M$ and $W$ are identified by the solutions to the equation $\mathcal{A}F=0$ ($\Pr{P}$-a.s.), subject to their respective terminal conditions:
\begin{align*}
\begin{pmatrix}
Z\\
M\\W
\end{pmatrix}(T)=
\begin{pmatrix}
z\circ\ell(T,\cdot)\\
f'\circ z\circ\ell(T,\cdot)\times {z\circ\ell(T,\cdot)}-{f\circ z\circ\ell(T,\cdot)}
\\
z\circ\ell(T,\cdot)\times\ell(T,\cdot)
\end{pmatrix}
\end{align*}
\end{corollary}
In practice, we are more interested in the type of $f$-divergence that gives the constant function $x\mapsto xf''(x)$. Such $f$-divergence allows us to solve $U$ and $V$ directly using path-dependent partial differential equations. 

\begin{proposition}\label{pro:52}
Suppose there exists $d\in(0,\infty)$ such that $xf''(x)=d$ for all $x\in\mathbb{R}_+$, and the function $f'$ diverges at infinity. In addition, the inverse function, $g:\mathsf{Im}f'\to(0,\infty)$, provides a twice-differentiable function $\mathbb{R}\ni x\mapsto g(x)\ind{x\in\mathsf{Im}f'}$. The value process and the worst-case risk, identified with the regular functionals  $U_t:=U(t,\cdot)$ and $V_t:=V(t,\cdot)$, solve the following path-dependent partial differential equations $\Pr{Q}_Z$-a.s.
\begin{align}\label{eq:functionalg}
\begin{split}
\mathcal{A}U_t+\frac{\theta}{2}
\frac{g''\bigl(\vartheta\left(U_t-c\right)\bigr)}{g'\bigl(\vartheta\left(U_t-c\right)\bigr)}
\left(\sigma_t\nabla_\omega U_t\right)^2 &=0\\
\mathcal{A}V_t+\frac{\vartheta g'\bigl(\vartheta\left(U_t-c\right)\bigr)\nabla_\omega U_t\,\sigma_t^2}{g\bigl(\vartheta\left(U_t-c\right)\bigr)}
  \,\nabla_\omega V_t&=0
\end{split}
\end{align}
subject to the terminal condition $U_T=V_T=\ell(T, \cdot)$. 
The cost process $\eta_t=\vartheta(V_t-U_t)$ for all $t\in[0,T]$. Defining $ I_c\coloneqq\{\vartheta^{-1}y+c\,|\,y\in\mathsf{Im}f'\}$, the solution exists if $g\bigl(\vartheta(\ell(T,\cdot)-c)\bigr)\ind{\ell(T,\cdot)\in I_c}$ is integrable for every $c\in\mathbb{R}$. 
\end{proposition}

\begin{proof} 
It follows from Corollary \ref{cor:46} that\footnote{\,We have shown in the proof of Proposition \ref{pro:45} that $f'$ diverges at infinity implies that $\mathsf{Im}f'$ is an open interval in the form of $(a,\infty)$. Then
\begin{align*}
U(t,\omega)=\vartheta^{-1}{f'(Z(t)(\omega))}+c> \vartheta^{-1}a+c\in I_c
\end{align*}
for all $\omega\in\{\omega\in\Omega\,|\,Z(t)(\omega)>0\}$. On the other hand,
for all $\omega\in\{\omega\in\Omega\,|\,Z(t)(\omega)=0\}$,
\begin{align*}
0=Z(t)(\omega)=&\, 
\E^{\Pr{Q}_Z}(Z(T)\ind{\ell(T,\cdot)\leqslant \vartheta^{-1}a+c}\,|\,\SgAlg{F}_t^0)(\omega)+
\E^{\Pr{Q}_Z}(Z(T)\ind{\ell(T,\cdot)> \vartheta^{-1}a+c}\,|\,\SgAlg{F}_t^0)\\
\geqslant&\,\E^{\Pr{Q}_Z}(g\left(\vartheta\left(\ell(T,\cdot)-c\right)\right)\ind{\ell(T,\cdot)> \vartheta^{-1}a+c}\,|\,\SgAlg{F}_t^0)(\omega)
\end{align*}
This implies that $
\E^{\Pr{Q}_Z}(\ind{\ell(T,\cdot)> \vartheta^{-1}a+c}\,|\,\SgAlg{F}_t^0)(\omega)=0$ 
by virtue of 
$\mathsf{Im}g=(0,\infty)$, 
which gives
\begin{align*}
U(t,\omega)=\E^{\Pr{Q}_Z}(\ell(T,\cdot)\,|\,\SgAlg{F}_t^0)(\omega)
=\E^{\Pr{Q}_Z}(\ell(T,\cdot)\ind{\ell(T,\cdot)\leqslant \vartheta^{-1}a+c}\,|\,\SgAlg{F}_t^0)(\omega)
\leqslant \vartheta^{-1}a+c\notin I_c
\end{align*}}
\begin{align*}
Z(t)=g\left(\vartheta\left(U_t-c\right)\right)\ind{Z(t)>0}
=g\left(\vartheta\left(U_t-c\right)\right)\ind{U_t\in I_c}
:=\mathfrak{g}\left(\vartheta\left(U_t-c\right)\right)
\end{align*}
for all $t\in[0,T]$, where $\mathfrak{g}$ denotes the twice-differentiable function 
$\mathbb{R}\ni x\mapsto g(x)\ind{x\in\mathsf{Im}f'}$. Since $(Z(t))_{t\in[0,T]}$ is a $\Pr{P}$-martingale that can be identified with a solution to the equation $\mathcal{A}F=0$ ($\Pr{P}$-a.s.), we have
\begin{align*}
0=\mathcal{A}\mathfrak{g}\left(\vartheta\left(U_t-c\right)\right)=
{\mathfrak{g}'\bigl(\vartheta\left(U_t-c\right)\bigr)}
\mathcal{A}U_t+\frac{\theta}{2}
{\mathfrak{g}''\bigl(\vartheta\left(U_t-c\right)\bigr)}
\left(\sigma_t\nabla_\omega U_t\right)^2 
\end{align*}
For all $\omega\in\Omega$ such that $U(t,\omega)\in I_c$, the equation is equivalent to\footnote{
For all $x\in(a,\infty)$, 
$\mathfrak{g}'(x)={g}'(x)>0$ (due to the convexity of $f$), and for all $x\in(-\infty,a\,]$, 
\begin{align*}
\mathfrak{g}'(x)=\lim_{h\to 0^-}\frac{\mathfrak{g}(x)-\mathfrak{g}(x-h)}{h}=0
\end{align*} 
Therefore,
$\mathfrak{g}(x)=g(x)\ind{x\in(a,\infty)}$ implies that 
$\mathfrak{g}'(x)=g'(x)\ind{x\in(a,\infty)}$, which in turns implies $\mathfrak{g}''(x)=g''(x)\ind{x\in(a,\infty)}$. 
For all $\omega\in\{\omega\in\Omega\,|\,U(t,\omega)\in I_c\}$, $\vartheta U(t,\omega)-c\in(a,\infty)$ and thus
\begin{align*}
\mathfrak{g}'\bigl(\vartheta\left(U(t,\omega)-c\right)\bigr)=
{g}'\bigl(\vartheta\left(U(t,\omega)-c\right)\bigr)>0
\qquad\text{and}\qquad 
\mathfrak{g}''\bigl(\vartheta\left(U(t,\omega)-c\right)\bigr)=
{g}''\bigl(\vartheta\left(U(t,\omega)-c\right)\bigr)
\end{align*} 
}
\begin{align}\label{eq:functionalU}
\mathcal{A}U_t+\frac{\theta}{2}
\frac{{g}''\bigl(\vartheta\left(U_t-c\right)\bigr)}
{{g}'\bigl(\vartheta\left(U_t-c\right)\bigr)}
\left(\sigma_t\nabla_\omega U_t\right)^2 =0
\end{align}
Noticing that $\{\omega\in\Omega\,|\,U(t,\omega)\in I_c\}$ has measure one under $\Pr{Q}_Z$\footnote{\,$
\Pr{Q}_Z(U(t,\cdot)\in I_c)=
\Pr{Q}_Z(Z(t)>0)=
\E(Z(T)\ind{Z(t)>0})=
\E(Z(t)\ind{Z(t)>0})=
\E(Z(t))=1$
}, the equation above holds $\Pr{Q}_Z$-a.s.

It follows from Eq.~\ref{eq:sde1} that the $\Pr{P}$-martingale $(Z(t))_{t\in[0,T]}$ solves the SDE
\begin{align*}
dZ(t)=\mathcal{A}\mathfrak{g}\left(\vartheta\left(U_t-c\right)\right)dt
+
\sigma_t\nabla_\omega \mathfrak{g}\left(\vartheta\left(U_t-c\right)\right)dW(t)
=\vartheta \mathfrak{g}'\left(\vartheta\left(U_t-c\right)\right)\sigma_t\nabla_\omega U_t\,dW(t)
\end{align*}
We may define a process $(Y(t))_{t\in[0,T]}$ by the stochastic integral
\begin{align*}
Y(t):=\int_0^t\left(\frac{\vartheta  {g}'\left(\vartheta\left(U_s-c\right)\right)}{ {g}\left(\vartheta\left(U_s-c\right)\right)}
\ind{U_s\in I_c}
 \sigma_t\nabla_\omega U_s\right)dW(s)
\end{align*}
for all $t\in[0,T]$. This transforms the SDE above into
\begin{align*}
dZ(t)=\vartheta {g}'\left(\vartheta\left(U_t-c\right)\right)\ind{U_t\in I_c}
\sigma_t\nabla_\omega U_t\,dW(t)
={g}\left(\vartheta\left(U_t-c\right)\right)\ind{U_t\in I_c}
dY(t)=Z(t)dY(t)
\end{align*}
suggesting that the process $\left(Z(t)\right)_{t\in[0,T]}$ is a Doleans-Dade exponent, i.e. $Z=\mathcal{E}\left(Y\right)$. 
Note that the SDE above ensures that $\left(Z(t)\right)_{t\in[0,T]}$ is a local martingale. To guarantee that it is indeed a martingale, we assume the Novikov's condition,
\begin{align*}
\E\left(\exp\left(\frac{1}{2}\int_0^T\left(\frac{\vartheta {g}'\left(\vartheta\left(U_t-c\right)\right)}{{g}\left(\vartheta\left(U_t-c\right)\right)}
\ind{U_t\in I_c} \sigma_t\nabla_\omega U_t\right)^2dt\right)\right)<\infty
\end{align*}
According to the Girsanov theorem, the Brownian motion under $\Pr{Q}_Z$ is given by adding an extra drift term. Noticing that $U_t\in I_c$ $\Pr{Q}_Z$-a.s., the Girsanov theorem transforms the SDE of the canonical process under $\Pr{P}$ (Eq.~\ref{eq:sde}) to the following SDE (in the sense that $\left(X(t)\right)_{t\in[0,T]}$ is a strong solution of the following under $\Pr{Q}_Z$), 
\begin{align}\label{eq:extradrift}
dX(t)=\left(\mu_t + \frac{\vartheta {g}'\left(\vartheta\left(U_t-c\right)\right)}{{g}\left(\vartheta\left(U_t-c\right)\right)}
 \sigma_t^2\,\nabla_\omega U_t\right)dt+\sigma_tdW^{\Pr{Q}_Z}(t)
\end{align}
The functional Ito formula, Eq.~\ref{eq:Adef}-\ref{eq:sde1}, applies to the alternative measure $\Pr{Q}_Z$ as well. Following the definition of the operator $\mathcal{A}$, we have
\begin{align*}
\mathcal{A}^{\Pr{Q}_Z}F(t,\cdot)=&\,
\mathcal{D}F(t,\cdot)+\left(\mu_t + \frac{\vartheta {g}'\left(\vartheta\left(U_t-c\right)\right)}{{g}\left(\vartheta\left(U_t-c\right)\right)}
\nabla_\omega U_t\,\sigma_t^2\right)
\nabla_\omega F(t,\cdot)+\frac{1}{2}\Tr\bigl(\sigma(t)^2\nabla_\omega^2F(t,\cdot)\bigr)
\\
=&\,\mathcal{A}F(t,\cdot) + \frac{ {g}'\left(\vartheta\sigma_t^2\left(U_t-c\right)\right)\nabla_\omega U_t\,\sigma_t^2}{{g}\left(\vartheta\left(U_t-c\right)\right)}
\nabla_\omega F(t,\cdot)
\end{align*}
for some regular functional $F:\Lambda_T^d\to\mathbb{R}$ and all $t\in[0,T]$.
The worst-case model risk, $\E^{\Pr{Q}_Z}\left(\ell(T,\cdot)\,|\,\SgAlg{F}_t\right)$, 
is a $\Pr{Q}_Z$-martingale. Identified with the regular functional $V$, it satisfies the following equation $\Pr{Q}_Z$-a.s. 
\begin{align}\label{eq:functionalV}
0=\mathcal{A}^{\Pr{Q}_Z}V_t=\mathcal{A}V_t+\frac{\vartheta g'\bigl(\vartheta\left(U_t-c\right)\bigr)\nabla_\omega U_t\,\sigma_t^2}{g\bigl(\vartheta\sigma_t\left(U_t-c\right)\bigr)}\,
  \nabla_\omega V_t
\end{align}
 
Combined with the terminal condition $U_T=V_T=\ell(T,\cdot)$, 
Eq.~\ref{eq:functionalU} and Eq.~\ref{eq:functionalV} provide the path-dependent partial differential equations that govern the value process and the worst-case risk, respectively. 
It follows from Proposition \ref{pro:45} that the solution indeed exists if $g\bigl(\vartheta(\ell(T,\cdot)-c)\bigr)\ind{\ell(T,\cdot)\in I_c}$ is integrable for every $c\in\mathbb{R}$. 
\end{proof}

The renowned Kullback-Leibler divergence provides us with much convenience on applying 
Proposition \ref{pro:52} into practice. The function $f'(x)=\ln x+1$ diverges at $\infty$, and its inverse $g:\mathbb{R}\to(0,\infty)$ given by $g(x)=e^{x-1}$ is twice-differentiable. In addition, the worst-case martingale density $Z(T)=e^{\vartheta(\ell(T,\cdot)-c)-1}>0$ supplies a measure $\Pr{Q}_Z$ that is equivalent to the reference measure $\Pr{P}$. 
Combining Corollary \ref{co:47} with Proposition \ref{pro:52}, and substituting $g(x)=e^{x-1}$ into Eq.~\ref{eq:functionalg}, we get the following corollary that applies to the Kullback-Leibler divergence. 
\begin{corollary}\label{cor:53}
Under the Kullback-Leibler divergence, suppose $\E\left(e^{\vartheta\ell(T,\cdot)}\right)<\infty$. Then there exists an unique solution to the problem of model risk quantification. The value process and the worst-case risk, identified with regular functionals $U_t:=U(t,\cdot)$ and $V_t:=V(t,\cdot)$, solve the following path-dependent partial differential equations $\Pr{P}$-a.s.
\begin{align}\label{eq:functionalU}
\begin{split}
\mathcal{A}U_t+\frac{\vartheta}{2}\left(\sigma_t\nabla_\omega U_t\right)^2 &=0\\
\mathcal{A}V_t+{\vartheta}\nabla_\omega U_t\,\sigma_t^2\, \nabla_\omega  V_t&=0
\end{split}
\end{align}
subject to the terminal condition $U_T=V_T=\ell(T, \cdot)$. 
The cost process $\eta_t=\vartheta(V_t-U_t)$ for all $t\in[0,T]$.
\end{corollary}

In practice, the path-dependent partial differential equations, Eq.~\ref{eq:functionalU}, are generally difficult to solve. However, we may convert Eq.~\ref{eq:functionalU} into normal non-linear partial differential equations for a special type of path dependency, formulated by
\begin{align}\label{eq:special}
\ell(T,\cdot)=h_0(T,X(T))+\int_0^T h_1(t,X(t))dt +\int_0^Th(t,X(t))dX(t)
\end{align}
for some functions $h:[0,T]\times\mathbb{R}^d\to\mathbb{R}^d$ and $h_i:[0,T]\times\mathbb{R}^d\to\mathbb{R}~(i=1,2)$. 
We further restrict the canonical process $X$ to the class of Ito diffusions. This means that the process is Markovian, and there exist functions $\mu:[0,T]\times\mathbb{R}^d\to\mathbb{R}^d$ and $\sigma:[0,T]\times\mathbb{R}^d\to\mathbb{R}^{d^2}$ such that $\mu_t=\mu(t,X(t))$ and $\sigma_t=\sigma(t,X(t))$. 
The path-dependent partial differential equations, Eq.~\ref{eq:functionalU}, degenerates to normal partial differential equations.

\begin{corollary} \label{cor:54}
Under the Kullback-Leibler divergence, suppose $\E\left(e^{\vartheta\ell(T,\cdot)}\right)<\infty$, the canonical process $(X(t))_{t\in[0,T]}$ solves the SDE, $dX(t)=\mu(t,X(t))dt+\sigma(t,X(t))dW(t)$, and the cumulative loss $\ell(T,\cdot)$ takes the form of Eq.~\ref{eq:special}. 
If there exists a function $\tilde{u}:[0,T]\times\mathbb{R}^d\to\mathbb{R}$ that solves the partial differential equation 
\begin{align}\label{eq:nonlinearpde1}
\begin{split}
\frac{\partial \tilde{u}(t,x)}{\partial t}+\mu(t,x)\left(\frac{\partial \tilde{u}(t,x)}{\partial x}+h(t,x)\right)+&\,\frac{\vartheta}{2}\left(\sigma(t,x)\left(\frac{\partial \tilde{u}(t,x)}{\partial x}+h(t,x)\right)\right)^2\\
+&\,\frac{1}{2}\Tr\left({\sigma(t,x)^2}\frac{\partial^2 \tilde{u}(t,x)}{\partial x^2}\right)+h_1(t,x)=0
\end{split}
\end{align}
and a function $\tilde{v}:[0,T]\times\mathbb{R}^d\to\mathbb{R}$ that solves the partial differential equation
\begin{align}\label{eq:nonlinearpde2}
\begin{split}
\frac{\partial \tilde{v}(t,x)}{\partial t}
+&\,\vartheta\left(\frac{\partial \tilde{u}(t,x)}{\partial x}+h(t,x)\right)\sigma(t,x)^2\left(\frac{\partial \tilde{v}(t,x)}{\partial x}+h(t,x)\right)
\\
+&\,\mu(t,x)\left(\frac{\partial \tilde{v}(t,x)}{\partial x}+h(t,x)\right)
+\frac{1}{2}\Tr\left({\sigma(t,x)^2}\frac{\partial^2 \tilde{v}(t,x)}{\partial x^2}\right)+h_1(t,x)\,=0
\end{split}
\end{align}
subject to the terminal condition $
\tilde{u}(T, \cdot)=\tilde{v}(T, \cdot)=h_0(T, \cdot)$, 
then the value process, the worst-case risk and the cost process, identified with regular functionals, follow
\begin{align*}
\begin{split}
{U}_t=&\,\tilde{u}(t, X(t))+\int_0^t h_1(s,X(s))ds +\int_0^th(s,X(s))dX(s)\\
{V}_t=&\,\tilde{v}(t, X(t))+\int_0^t h_1(s,X(s))ds +\int_0^th(s,X(s))dX(s)
\end{split}
\end{align*}
and $\eta_t=\vartheta\bigl(\tilde{v}(t, X(t))-\tilde{u}(t, X(t))\bigr)$ for all $t\in[0,T]$.
\end{corollary}

\begin{proof}
We first define regular functionals $\tilde{U},\tilde{V}:\Lambda_T^d\to\mathbb{R}$ by
\begin{align}\label{eq:split}
\begin{split}
\tilde{U}_t:=&\,U_t-\int_0^t h_1(s,X(s))ds -\int_0^th_2(s,X(s))dX(s)\\
\tilde{V}_t:=&\,V_t-\int_0^t h_1(s,X(s))ds -\int_0^th_2(s,X(s))dX(s)
\end{split}
\end{align}
The horizontal and vertical derivatives can be derived from Eq.~\ref{eq:split},
\begin{align*}
\begin{split}
\mathcal{D}\tilde{U}_t=\mathcal{D}U_t-h_1(t,X(t))
\qquad\text{and}\qquad &\,
\mathcal{D}\tilde{V}_t=\mathcal{D}V_t-h_1(t,X(t))\\
\nabla_\omega\tilde{U}_t=\nabla_\omega U_t-h_2(t,X(t))
\qquad\text{and}\qquad &\,
\nabla_\omega\tilde{V}_t=\nabla_\omega V_t-h_2(t,X(t))\\
\nabla_\omega^2\tilde{U}_t=\nabla_\omega^2 U_t
\qquad\text{and}\qquad &\,
\nabla_\omega^2\tilde{V}_t=\nabla_\omega^2 V_t
\end{split}
\end{align*}
Substituting the equations above into Eq.~\ref{eq:functionalU}, we transform Eq.~\ref{eq:functionalU} to
\begin{align}\label{eq:functional1}
\begin{split}
\mathcal{D}\tilde{U}_t+\mu(t,X(t))\bigl(\nabla_\omega\tilde{U}_t+h(t,X(s))\bigr)+&\,\frac{\vartheta}{2}\left(\sigma(t,X(t))\bigl(\nabla_\omega\tilde{U}_t+h(t,X(s))^2\bigr)\right)^2\\
+&\,
\frac{1}{2}\Tr\left(\sigma(t,X(t))^2\,\nabla_\omega^2\tilde{U}_t\right)+
h_1(t,X(t)) =0
\end{split}
\end{align}
and
\begin{align}\label{eq:functional2}
\begin{split}
\mathcal{D}\tilde{V}_t\,+&\,{\vartheta}\bigl(\nabla_\omega\tilde{U}_t+h(t,X(s))^2\bigr)\sigma(t,X(t))^2\bigl(\nabla_\omega\tilde{V}_t+h(t,X(s))^2\bigr)\\
+&\,\mu(t,X(t))\bigl(\nabla_\omega\tilde{V}_t+h(t,X(s))\bigr)+
\frac{1}{2}\Tr\left(\sigma(t,X(t))^2\,\nabla_\omega^2\tilde{V}_t\right)+
h_1(t,X(t)) =0
\end{split}
\end{align}
If there exists a function $\tilde{u}:[0,T]\times\mathbb{R}^d\to\mathbb{R}$ that solves the partial differential equation
\begin{align*}
\frac{\partial \tilde{u}(t,x)}{\partial t}+\mu(t,x)\left(\frac{\partial \tilde{u}(t,x)}{\partial x}+h(t,x)\right)+&\,\frac{\vartheta}{2}\left(\sigma(t,x)\left(\frac{\partial \tilde{u}(t,x)}{\partial x}+h(t,x)\right)\right)^2\\
+&\,\frac{1}{2}\Tr\left({\sigma(t,x)^2}\frac{\partial^2 \tilde{u}(t,x)}{\partial x^2}\right)+h_1(t,x)=0
\end{align*}
and a function $\tilde{v}:[0,T]\times\mathbb{R}^d\to\mathbb{R}$ that solves
\begin{align*}
\frac{\partial \tilde{v}(t,x)}{\partial t}
+&\,\vartheta\left(\frac{\partial \tilde{u}(t,x)}{\partial x}+h(t,x)\right)\sigma(t,x)^2\left(\frac{\partial \tilde{v}(t,x)}{\partial x}+h(t,x)\right)
\\
+&\,\mu(t,x)\left(\frac{\partial \tilde{v}(t,x)}{\partial x}+h(t,x)\right)
+\frac{1}{2}\Tr\left({\sigma(t,x)^2}\frac{\partial^2 \tilde{v}(t,x)}{\partial x^2}\right)+h_1(t,x)\,=0
\end{align*}
then the regular functionals defined by $\tilde{U}_t:=\tilde{u}(t,X(t))$ and $\tilde{V}_t:=\tilde{v}(t,X(t))$, for all $t\in[0,T]$, satisfy Eqs.~\ref{eq:functional1} and \ref{eq:functional2}. 
The terminal condition $\tilde{U}_T=\tilde{V}_T=h_0(T,X(T))$ is satisfied if $\tilde{u}(T, x)=\tilde{v}(T, x)=h_0(T,x)$ holds for all $x\in\mathbb{R}$.
\end{proof}
Note that Eq.~\ref{eq:nonlinearpde1}-\ref{eq:nonlinearpde2} are non-linear parabolic partial differential equations and in general have to be solved numerically. 

\section{Concluding Remarks}
This paper provides a theoretical framework of formulating and solving the problem of model risk quantification in a path-dependent setting. We need several ingredients to formulate the problem, including terminal time $T$, a (path-dependent) loss function $\ell$, a nominal model (i.e. a canonical process $(X_t)_{t\in[0,T]}$ under a nominal measure $\Pr{P}$) and some $f$-divergence. The non-parametric nature of this approach relies on the $f$-divergence to restrict the set of proper alternative models. This is, however, only applicable to measures that are absolutely continuous w.r.t the nominal measure. More generic distance measure, such as the Wasserstein metric, may be applied instead \cite[]{third}. Despite of this incompleteness, $f$-divergence, especially the Kullback-Leibler divergence, is most tractable and yield simple results for path-dependent problems.


\bibliography{modelrisk}
\bibliographystyle{chicago}
\end{document}